\pdfoutput=1

\documentclass[11pt,a4paper,fleqn]{article}

\usepackage{amsthm}
\usepackage{amssymb}
\usepackage{amstext}
\usepackage{amsmath}
\usepackage{url}
\usepackage[margin = 40pt,font=small,labelfont=bf]{caption}
\usepackage{pdfsync}
\usepackage[mac]{inputenc}
\usepackage[T1]{fontenc}
\usepackage{color}
\usepackage{enumerate}
\usepackage{bbm}
\usepackage{ae,aecompl}
\usepackage{pdfpages}
\usepackage{epstopdf}
\usepackage{graphicx}
\usepackage{epsfig}
\usepackage{subfigure}
\usepackage{tabularx}
\usepackage{arydshln}
\usepackage{geometry}
\usepackage{fancyhdr,framed,latexsym,multicol,pstricks,slashed,xcolor}
\usepackage{algorithm,algorithmic,hhline,array,isomath}
\usepackage{appendix}
\usepackage{hyperref}

\makeatletter
\newcommand{\specificthanks}[1]{\@fnsymbol{#1}}
\makeatother

\newenvironment{frcseries}{\fontfamily{frc}\selectfont}{}

\newcommand{\iter}[1]{{#1}^{(\ell)}}
\newcommand{\iiter}[1]{{#1}^{(\ell+1)}}
\newcommand{\ter}[1]{{#1}^{(\ell-1)}}

\numberwithin{equation}{section}

\newcommand{\thefont}[2]{\fontsize{#1}{#2}\fontshape{n}\selectfont}
\newcommand{\1}{\rlap{\thefont{10pt}{12pt}1}\kern.16em\rlap{\thefont{11pt}{13.2pt}1}\kern.4em}

\theoremstyle{plain}
\newtheorem{prop}{Proposition}[section]

\newtheorem{lem}{Lemma}[section]

\theoremstyle{definition}
\newtheorem{defin}{Definition}[section]

\theoremstyle{remark}
\newtheorem{rem}{Remark}[section]

\newcommand{\R}{{\mathbb R}}

\newcommand{\WS}{W_2(\Omega)}

\newcommand{\WSd}{W_2(\R^{d})}
\newcommand{\WSac}{W_2^{ac}(\Omega)}
\newcommand{\LL}{L^{2}_{\mu_{r}}(\Omega)}
\newcommand{\VV}{V_{\mu_{r}}(\Omega)}

\newcommand{\CL}{\mathrm{CL}}

\newcommand{\LLemp}{L^{2}_{\bar{\bnu}}(\Omega)}
\newcommand{\U}{{\mathcal U}}
\newcommand{\spann}{\operatorname{Sp}}

\newcommand{\costn}{\operatorname{D}_W}
\newcommand{\bfun}{\boldsymbol{f}}

\newcommand{\spaceK}{E}
\newcommand{\id}{\text{id}}
\newcommand{\bnu}{\boldsymbol{\nu}}

\DeclareMathOperator{\Proj}{Proj}
\DeclareMathOperator{\Prox}{Prox}
\DeclareMathOperator{\argmin}{argmin}
\newcommand{\uargmin}[1]{\underset{#1}{\argmin}\;}
\DeclareMathOperator{\argmax}{argmax}

\newcommand{\umin}[1]{\underset{#1}{\min}\;}

\newcommand{\dotprod}[2]{\ensuremath{\langle #1 , #2\,\rangle}}

\title{Log-PCA versus Geodesic PCA of histograms in the\\ Wasserstein space.}

\author{Elsa Cazelles\textsuperscript{\specificthanks{5}}\thanks{Institut de Math\'ematiques de Bordeaux et CNRS, IMB-UMR5251, Universit\'e de Bordeaux} \and Vivien Seguy\textsuperscript{\specificthanks{5}}\thanks{Graduate School of Informatics, Kyoto University.} \and J\'er\'emie Bigot\footnotemark[1] \and  Marco Cuturi\thanks{CREST, ENSAE, Universit\'e  Paris-Saclay.}   \and Nicolas Papadakis\footnotemark[1]}

\begin{document}

\maketitle

\thispagestyle{empty}

\begin{abstract}
This paper is concerned by the statistical analysis of data sets whose elements are random histograms. For the purpose of learning  principal modes of variation from such data, we consider the issue of computing the PCA of histograms with respect to the 2-Wasserstein distance between probability measures. To this end, we propose to compare the methods of log-PCA and geodesic PCA in the Wasserstein space as introduced in \cite{BGKL13,NIPS2015_5680}.  Geodesic PCA involves solving a non-convex optimization problem. To solve it approximately, we propose a novel forward-backward algorithm. This allows a detailed comparison between log-PCA and geodesic PCA of one-dimensional histograms, which we carry out using various datasets, and stress the benefits and drawbacks of each method. We extend these results for two-dimensional data and compare both methods in that setting.

\end{abstract}


\let\oldthefootnote=\thefootnote
\renewcommand{\thefootnote}{{\specificthanks{5}}}
\footnotetext[5]{These authors contributed equally.}
\let\thefootnote=\oldthefootnote
\section{Introduction}
Most datasets describe multivariate data, namely vectors of relevant features that can be modeled as random elements sampled from an unknown distribution. In that setting, Principal Component Analysis (PCA) is certainly the simplest and most widely used approach to reduce the dimension of such datasets. We consider in this work the statistical analysis of data sets whose elements are histograms supported on the real line. Just as with PCA, our main goal in that setting is to compute the principal modes of variation of histograms around their mean element and therefore facilitate the visualization of such datasets. However, since the number, size or locations of significant bins in the histograms of interest may vary from one histogram to another, using standard PCA on histograms (with respect to the Euclidean metric) is bound to fail (see for instance Figure~\ref{eucli_wass_gaussian}).

In this paper, we propose to use the 2-Wasserstein metric~\cite[\S7.1]{villani-topics} to measure the distance between histograms, and to compute their modes of variation accordingly. In our approach, histograms are seen as piecewise constant probability density functions (pdf) supported in a given interval $\Omega$ of the real line. In this setting, the variability in a set of histograms can be analyzed via the notion of Geodesic PCA (GPCA) of probability measures in the Wasserstein space $\WS$ admitting these histograms as pdf. That approach has been recently proposed in the statistics and machine learning literature in~\cite{BGKL13} for probability measures supported on the real line, and in~\cite{NIPS2015_5680,Wang2013} for discrete probability measures on $\mathbb{R}^d$. However, implementing GPCA remains a challenging computational task even in the simplest case of pdf's supported on $\R$. The purpose of this paper is to provide a fast algorithm to perform GPCA of probability measures supported on the real line, and to compare its performances with log-PCA, namely standard PCA in the tangent space at the Wasserstein barycenter of the data \cite{geodesicPCA,petersen2016}.\\ 

\begin{figure}
\centering
\includegraphics[width=0.8\textwidth,height=0.55\textwidth]{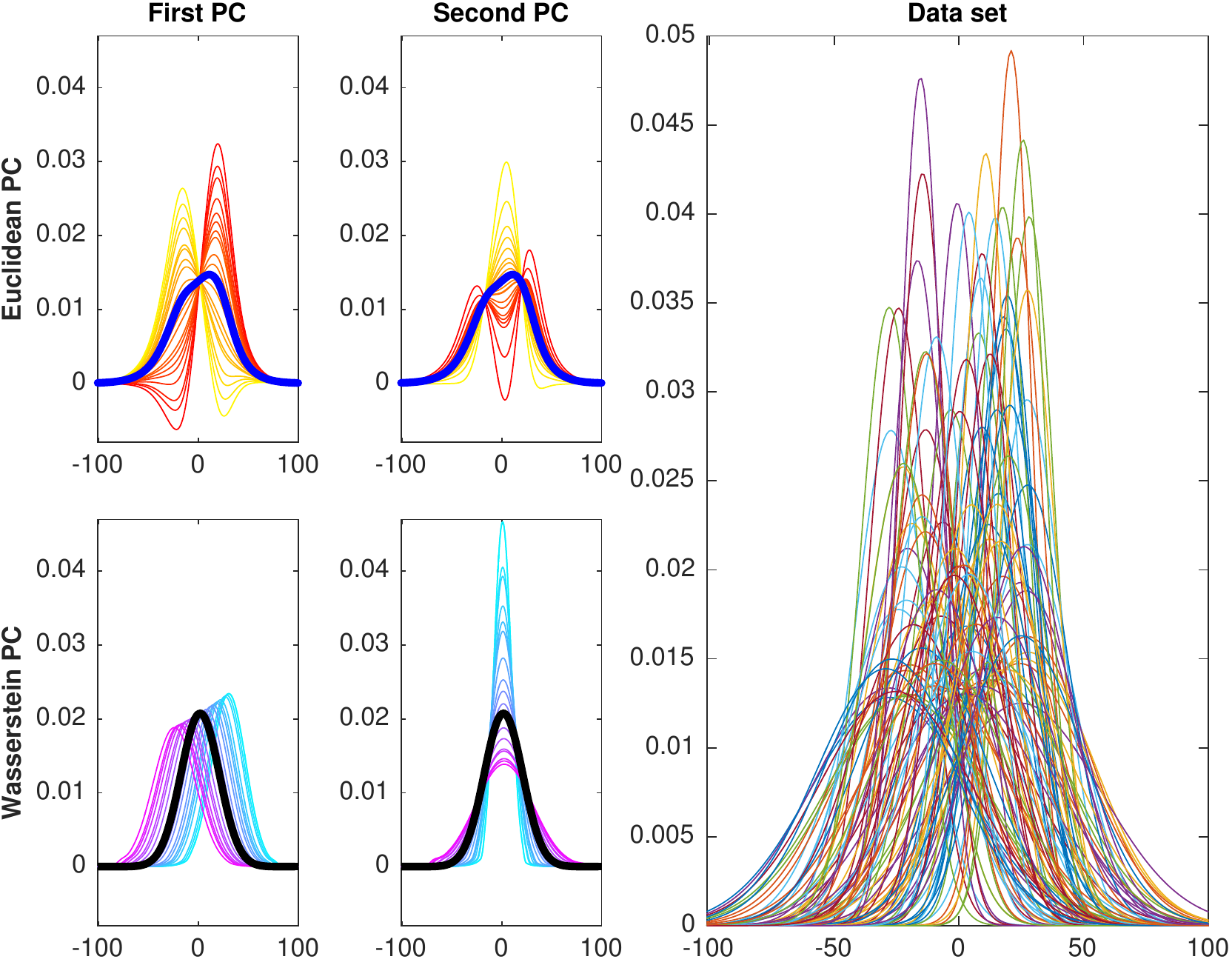}
\caption{Synthetic example. (Right) A data set of $n=100$ Gaussian histograms randomly translated and scaled. (Top-left) Standard PCA of this data set with respect to the Euclidean metric. The Euclidean barycenter of the data set is depicted in blue. (bottom-left) Geodesic PCA with respect to the Wasserstein metric using the iterative geodesic algorithm \eqref{eq:optimCPCA2}. The black curve represents the density of the Wasserstein barycenter. Colors encode the progression of the pdf of principal geodesic components in $\WS$. }
\label{eucli_wass_gaussian}
\end{figure}

\subsection{Related results}\label{sec:related}
\paragraph{Foundations of Geodesic PCA in the Wasserstein space.} The space of probability measures (with finite second moment) endowed with the 2-Wasserstein distance is not a Hilbert space. Therefore, standard PCA, which involves computing a covariance matrix, cannot be applied directly to compute principal mode of variations in a Wasserstein sense. Nevertheless, a meaningful notion of PCA can still be defined by relying on the pseudo-Riemannian structure of the Wasserstein space, which was extensively studied in~\cite{ambrosio2004gradient} and~\cite{Ambrosio2006}. Following this principle, a framework for GPCA of probability measures supported on a interval $\Omega \subset \R$ was introduced in~\cite{BGKL13}. GPCA is defined as the problem of estimating a principal geodesic subspace (of a given dimension) which maximizes the variance of the projection of the data to that subspace. In that approach the base point of that subspace is the Wasserstein barycenter of the data as introduced in~\cite{MR2801182}, which is also known as a Fr\'echet mean. Existence, consistency and a detailed characterization of GPCA in $\WS$ were studied in~\cite{BGKL13}. In particular, the authors have shown that this approach is equivalent to map the data in the tangent space of  $\WS$ at the Fr\'echet mean, and then to perform a PCA in this Hilbert space that is constrained to lie in a convex and closed subset of functions. Mapping the data to this tangent space is not difficult in the one-dimensional case as it amounts to computing a set of optimal maps between the data and their Wasserstein barycenter, for which a closed form is available using their quantile functions (see for example~\cite[\S2.2]{villani-topics}). To perform PCA on the mapped data,~\cite{BGKL13} fell short of proposing an algorithm to minimize that problem, which has a non-convex and non-differentiable objective function as well as involved constraints. Only a numerical approximation to the computation of GPCA was proposed in~\cite{BGKL13}, which amounts to applying log-PCA, namely a standard PCA of the dataset mapped beforehand to the tangent space of $\WS$ at its Fr\'echet mean. 

\paragraph{Previous work in the one-dimensional case.} PCA of histograms with respect to the Wasserstein metric has also been proposed in~\cite{Verde2015} in the context of symbolic data analysis. Their approach consists in computing a standard PCA in the Hilbert space $L^{2}([0,1])$ of the quantile functions associated to the histograms. Therefore, the algorithm in~\cite{Verde2015} corresponds to log-PCA of probability measures as suggested in~\cite{BGKL13}, but it does not solve the problem of convex-constrained PCA in a Hilbert space associated to an exact GPCA in $\WS$.

\paragraph{PGA and log-PCA in Riemannian manifolds} The method of GPCA proposed in~\cite{BGKL13} clearly shares similarities with analogs of PCA for data belonging to a Riemannian manifold $\mathcal{M}$ of finite dimension. These methods, generally referred to as Principal Geodesic Analysis (PGA) extend the notion of classical PCA in Euclidean spaces for the purpose of analyzing data belonging to curved Riemannian manifolds (see e.g.~\cite{geodesicPCA,exactPGA}). This generalization of PCA proceeds  by replacing Euclidean concepts of vector means, lines and orthogonality by the more general notions in Riemannian manifolds of Fr\'echet mean, geodesics, and orthogonality in tangent spaces.

In~\cite{geodesicPCA}, linearized PGA, which we refer to as log-PCA, is defined as follows. In a first step, data are mapped to  the tangent space $T_{\bar{x}}\mathcal{M}$ at their Fr\'echet mean $\bar{x}$ by applying the logarithmic map $\log_{\bar{x}}$ to each data point. Then, in a second step, standard PCA in the Euclidean space $T_{\bar{x}}\mathcal{M}$ can be applied. This provides a family of orthonormal tangent vectors.  Principal components of variation in $\mathcal{M}$ can then defined by back-projection of these tangent vectors on $\mathcal{M}$  by using the exponential map at $\bar{x}$, that is known to parameterize geodesics at least locally. Log-PCA (e.g.\ linearized PGA) has low computational cost, but this comes at the expense of two simplifications and drawbacks:

\begin{description}
\item[(1)] First, log-PCA amounts to substituting geodesic distances between data points by the linearized distance in $T_{\bar{x}}\mathcal{M}$, which may not always be a good approximation because of the curvature of $\mathcal{M}$, see e.g.~\cite{exactPGA}.
\item[(2)]  Secondly, the exponential map at the Fr\'echet mean parameterizes geodesics only locally, which implies that principal components in $\mathcal{M}$ obtained with log-PCA may not be geodesic along the typical range of the dataset.
\end{description}

\paragraph{Numerical approaches to GPCA and log-PCA in the Wasserstein space.} Computational methods have been introduced in~\cite{NIPS2015_5680,Wang2013} to extend the concepts of PGA in a Riemannian manifold to that of the space $\WSd$ of  probability measures supported on $\R^{d}$ endowed with the Wasserstein metric.~\cite{Wang2013} propose to compute a notion of template measure (using $k$-means clustering) of a set of (possibly discrete) probability measures, and to consider then the optimal transport plans from each measure in the data set to  that template measure. Computation of the barycentric projection of each optimal transport plan leads to a set of Monge maps over which a standard PCA can be applied, resulting in an orthonormal family of tangent vectors defined on the support of the template measure. Principal components of variation in $\R^{d}$ can then be obtained through the push-forward operator, namely by moving the mass along these tangent vectors. This approach, analog to log-PCA on Riemannian manifolds, suffers from the main drawbacks  mentioned above: For $d>1$, the linearized Wasserstein distance may be a crude approximation of the Wasserstein distance, and there is no guarantee that the computed tangent vectors parameterize geodesics of sufficient length to summarize most of the variability in the dataset. Losing geodesicity means that the principal components are curves in $\WSd$ along which the mass may not be transported  optimally. Therefore, this may significantly reduce the interpretability of these principal components. A different approach was proposed in~\cite{NIPS2015_5680}, in which the notion of generalized geodesics in $\WSd$ (see e.g.\ Chapter 9 in~\cite{Ambrosio2006}) is used to define a notion of PGA of discrete probability measures. In~\cite{NIPS2015_5680}, generalized geodesics are parameterized using two velocity fields defined on the support of the Wasserstein barycenter. The authors proposed to minimize directly the distances from the measures in the dataset to these generalized geodesics, by updating these velocity fields which are constrained to be in opposite directions.
 This  approach is more involved computationally than log-PCA, but it avoids some of the drawbacks highlighted above. Indeed, the resulting principal components yield curves in $\WSd$ that are guaranteed to be approximately geodesics. Nevertheless, the computational method in~\cite{NIPS2015_5680} requires several numerical approximations (both on the optimal transport metric and on geodesics) for the algorithm to work at large scales. Therefore, it does not solve exactly the problem of computing geodesic PCA in $\WSd$.

\subsection{Main contributions}

In this paper, we focus on computing an exact GPCA on probability measures supported on $\Omega \subset \R$. The case $d=1$ has the advantage that the linearized Wasserstein distance in the tangent space is equal to the Wasserstein distance in the space $\WS$. The main challenge is thus to obtain principal curves which are geodesics along the range of the dataset.
 
The first contribution of this paper is to propose two fast algorithms for GPCA in $\WS$. The first algorithm finds iteratively geodesics such that the Wasserstein distance between the dataset and the parameterized geodesic is minimized with respect to $\WS$. This approach is thus somewhat similar to the one in~\cite{NIPS2015_5680}. However, a heuristic barycentric projection is used in~\cite{NIPS2015_5680} to remain in the feasible set of constraints during the optimization process. In our approach, we rely on proximal operators of both the objective function and the constraints to obtain an algorithm which is guaranteed to converge to a critical point of the objective function. Moreover, we show that the global minimum of our objective function for the first principal geodesic curve corresponds indeed to the solution of the exact GPCA problem defined in~\cite{BGKL13}. While this algorithm is able to find iteratively orthogonal principal geodesics, there is not guarantee that several principal geodesics parameterize a surface which is also geodesic. This is why we propose a second algorithm which computes all the principal geodesics at once by parameterizing a geodesic surface as a convex combination of optimal velocity fields, by relaxing the orthogonality constraint between principal geodesics. Both algorithms are variant of the proximal Forward-Backward algorithm. They converge to a stationary point of the objective function, as shown by recent results in non-convex optimization based on proximal methods~\cite{Attouch,Ochs}.
 
Our second contribution is to numerically compare log-PCA in $\WS$ as done in~\cite{BGKL13} (for $d=1$) or~\cite{Wang2013}, with our approach to solve the exact Wasserstein GPCA problem. 

Finally, we discuss some extensions of these results to the comparison of log-PCA and geodesic PCA of two-dimensional histograms.

\subsection{Structure of the paper}

In Section~\ref{sec:back}, we provide some background  on GPCA in the Wasserstein space $\WS$, borrowing  material from previous work in~\cite{BGKL13}. Section~\ref{sec:linear_pca} describes log-PCA in $\WS$, and some of its limitations are discussed. Section~\ref{sec:reformulation} contains the main results of the paper, namely two algorithms for computing GPCA. In Section~\ref{sec:num}, we provide a comparison between GPCA and log-PCA using statistical analysis of real datasets of histograms.  Section~\ref{sec:end} contains a discussion on a comparison of GPCA and log-PCA of histograms supported on $\R^{2}$. Some perspectives on this work are also given. Finally, various details on the implementation of the algorithms are deferred to technical Appendices.


\section{Background on Geodesic PCA in the Wasserstein space} \label{sec:back}

\subsection{Definitions and notations}

Let $\Omega$ be a (possibly unbounded) interval in $\R$. Let $\nu$ be a probability measure (also called distribution) over $(\Omega,{\cal B}(\Omega))$ where ${\cal B}(\Omega)$ is the $\sigma$-algebra of Borel subsets of $\Omega$.  For a mapping $T : \Omega \to \Omega$,  the  push-forward measure $ T  \# \nu $ is a probability measure on $\Omega$ defined by $( T  \# \nu  )(A) =  \nu\{x\in\Omega|T(x)\in A\}$,  for any $A \in {\cal B}(\Omega)$. The cumulative distribution function (cdf) and the (generalized) quantile function of $\nu$ are denoted respectively by $F_\nu$ and $F_\nu^{-}$. The Wasserstein space $\WS$ is the set of probability measures with support included in $\Omega$ and having a finite second moment, that is endowed with the quadratic Wasserstein distance $d_W$ defined by
\begin{equation} \label{eq:wdist}
d_W^2(\mu,\nu) := \int_{0}^1 (F_{\mu}^{-}(\alpha) - F_{\nu}^{-}(\alpha))^2 d\alpha ,\; \mu, \nu \in\WS.
\end{equation}
We also denote by $\WSac$ the set of measures $\nu \in \WS$ that are absolutely continuous with respect to the Lebesgue measure $dx$ on $\R$. If $\mu \in \WSac$  then $T^{\ast} = F_\nu^{-} \circ F_\mu$ will be referred to as the optimal mapping to push-forward $\mu$ onto $\nu$ and in this case $d_W^2(\mu,\nu)= \int_\Omega(T^{\ast}(x)-x)^2d\mu(x)$. For a detailed analysis of $\WS$ and its connection with optimal transport theory, we refer to~\cite{villani-topics}.

\subsection{The pseudo Riemannian structure of the Wasserstein space}

In what follows, $\mu_{r}$ denotes a reference measure in $\WSac$, whose choice will be discussed later on. The space $\WS$  has a formal Riemannian structure described, for example, in~\cite{ambrosio2004gradient}.  The tangent space at $\mu_{r}$ is defined as the Hilbert space $\LL$ of real-valued, $\mu_{r}$-square-integrable functions on $\Omega$, equipped with the  inner product $\langle \cdot ,\cdot\rangle_{\mu_{r}} $ defined by
$
\langle u ,v \rangle_{\mu_{r}} = \int_{\Omega} u(x) v(x) d \mu_{r} (x), \; u,v \in \LL,
$
and associated norm $\| \cdot \|_{\mu_{r}}$. We define  the  exponential and the logarithmic maps at $\mu_{r}$, as follows.

\begin{defin} \label{def:explog}
Let ${\rm id} : \Omega \to \Omega$ be the identity mapping. The exponential $\exp_{\mu_{r}}: \LL \to W_2(\mathbb{R})$ and logarithmic $\log_{\mu_{r}}: \WS \to \LL$ maps are defined respectively as
\begin{equation}\label{eq:exp}
 \exp_{\mu_{r}}(v)=({\rm id}+v) \# \mu_{r}   \quad \text{ and } \quad \log_{\mu_{r}}(\nu)=F_\nu^{-} \circ F_{\mu_{r}} - {\rm id}.
\end{equation}
\end{defin}

Contrary to  the setting of Riemannian manifolds, the ``exponential map'' $\exp_{\mu_{r}}$ defined above is not  a local homeomorphism from a neighborhood of the origin in the ``tangent space'' $\LL$ to the space $\WS$, see e.g.~\cite{ambrosio2004gradient}. Nevertheless, it is shown in~\cite{BGKL13} that $\exp_{\mu_{r}}$ is an isometry when  restricted to the following specific set of functions
$$
\VV:=\log_{\mu_{r}}(\WS) = \left\{  \log_{\mu_{r}}(\nu) \; ; \; \nu \in \WS \right\}  \subset \LL,
$$
and that the following results hold (see~\cite{BGKL13}).
\begin{prop}\label{prop_vmur}The subspace  $\VV$ satisfies the following properties :  \label{prop:exp} 
\begin{description}
\item[(P1)] the exponential map $\exp_{\mu_{r}}$ restricted to $\VV$ is an isometric homeomorphism, with inverse $\log_{\mu_{r}}$. We have hence $W_2(\nu,\eta) = ||\log_{\mu_{r}}(\nu)-\log_{\mu_{r}}(\eta)||_{\LL}$.
\item[(P2)] the set $\VV:=\log_\mu(\WS)$ is closed and convex in $\LL$.
\item[(P3)] the space $\VV$ is the set of functions $v \in \LL$ such that $T:={\rm id}+v$ is $\mu_{r}$-almost everywhere non decreasing and that $T(x) \in \Omega$, for $x \in \Omega$.
\end{description}
\end{prop}

Moreover, it follows, from~\cite{BGKL13}, that geodesics in $\WS$ are exactly the image under $\exp_{\mu_{r}}$ of straight lines in $\VV$. This property is stated in the following lemma.

\begin{lem}\label{lem:gamma}
Let $\gamma : [0,1] \to \WS$ be a curve and let $v_0 := \log_{\mu_{r}}(\gamma(0))$, $v_1 := \log_{\mu_{r}}(\gamma(1))$. Then $\gamma = (\gamma_{t})_{t \in [0,1]}$ is a geodesic   if and only if $\gamma_{t} = \exp_{\mu_{r}}((1-t)v_0+tv_1)$, for all $t \in [0,1]$.
\end{lem}

\subsection{GPCA for probability measures}

Let $\nu_1,\ldots,\nu_n$ be a set of probability measures in $\WSac$.  Assuming that each $\nu_i$ is absolutely continuous simplify the following presentation, and it is in line with the purpose of   statistical analysis of histograms. We define now the notion of (empirical) GPCA of this set of probability measures by following the approach in~\cite{BGKL13}. The first step is to choose the reference measure $\mu_{r}$. To this end, let us introduce the Wasserstein barycenter~\cite{MR2801182} or Fr\'echet mean of the  $\nu_i$'s, that is defined as the probability measure $\bar{\bnu}$,
$$
\bar{\bnu} = \argmin_{\mu \in \WS} \frac{1}{n} \sum_{i=1}^{n} d_W^{2}(\nu_{i},\mu).
$$
Note that it immediately follows from results in~\cite{MR2801182} that $\bar{\bnu} \in \WSac$, and that its cdf satisfies
\begin{equation}
F_{\bar{\bnu}}^{-} = \frac{1}{n} \sum_{i=1}^{n} F^{-}_{\nu_{i}}. \label{eq:quant}
\end{equation}
A typical choice for the reference measure is to take $\mu_{r} = \bar{\bnu}$ which represents an average location in the data around which can be computed the principal sources of geodesic variability. To introduce the notion of a  principal geodesic subspace of the measures $\nu_1,\ldots,\nu_n$, we need to introduce further notation and definitions. Let $G$ be a subset of $\WS$. The distance between $\mu \in \WS$ and the set $G$ is
$
d_W(\nu,{G}) = \inf_{\lambda \in {G}}d_W(\nu,\lambda),
$
and  the average distance between the data and $G$ is taken as
\begin{equation}
\label{eq:avg_dist}
\costn(G) := \frac{1}{n} \sum_{i=1}^{n} d_W^2(\nu_{i},G).
\end{equation}
\begin{defin}
Let $K$ be some positive integer. A subset $G \subset \WS$ is said to be a geodesic set of dimension $\dim(G) = K$ if $\log_{\mu_{r}}(G)$ is a convex set such that the dimension of the smallest affine subspace of $\LL$ containing $\log_{\mu_{r}}(G)$ is of dimension $K$.
\end{defin}

The notion of  principal geodesic subspace (PGS)   with respect to the reference measure $\mu_{r} = \bar{\bnu}$ can now be presented below.

\begin{defin}\label{def:GPCA}
Let $\CL(W)$ be the metric space  of  nonempty, closed subsets of $\WS$, endowed with the Hausdorff distance, and
$$
\mathrm{CG}_{\bar{\bnu},K}(W) = \left\{  G \in \CL(W) \; | \; \bar{\bnu} \in G, \; \mbox{$G$ is a geodesic set and $\dim(G) \leq K$}   \right\},\; K\ge1.
$$
A principal geodesic subspace (PGS) of $\bnu$ of dimension $K$ with respect to $\bar{\bnu}$ is a set
\begin{equation}
G_K\in \uargmin{G \in \mathrm{CG}_{\bar{\bnu},K}(W)} \costn(G). \label{eq:hatGk}
\end{equation}
\end{defin}
When $K=1$ , searching for the first PGS of $\bnu$  simply amounts to search for a geodesic curve $\gamma^{(1)}$ that is a solution of the following optimization problem:
$$
\tilde{\gamma}^{(1)}:= \uargmin{\gamma} \left\{  \frac{1}{n} \sum_{i=1}^{n} d_W^2(\nu_{i},\gamma) \; | \; \mbox{ $\gamma$ is a geodesic in $\WS$ passing through  $\mu_{r} = \bar{\bnu}$. } \right\}.
$$
We remark that this definition of $\tilde{\gamma}^{(1)}$ as the first principal geodesic curve of variation in $\WS$ is consistent with the usual concept of PCA in a Hilbert space in which geodesic are straight lines.

For a given dimension $k$, the GPCA problem consists in finding a nonempty closed geodesic subset of dimension $k$ which contains the reference measure $\mu_{r}$ and minimizes Eq. \eqref{eq:avg_dist}. We describe in the next section how we can parameterize such sets $G$.

\subsection{Geodesic PCA parameterization}

GPCA can be formulated as an optimization problem in the Hilbert space $\LLemp$. To this end, let us define the functions $\omega_{i} = \log_{\bar{\bnu}}(\nu_{i})$ for $1 \leq i \leq n$ that corresponds to the data mapped in the tangent space. It can be easily checked that this set of functions is centered in the sense that $\frac{1}{n} \sum_{i = 1}^{n} \omega_{i} = 0$. Note that, in a one-dimensional setting, computing $\omega_{i}$ (mapping of the data to the tangent space) is straightforward since the optimal maps $T^{\ast}_{i} = F_{\nu_{i}}^{-} \circ F_{\bar{\bnu}}$  between the  data and their Fr\'echet mean are available in a simple and closed form.
 
 For $\U=\{u_1,\ldots,u_K\}$ a collection of $K\geq 1$ functions belonging to $\LLemp$, we  denote by $\spann(\U)$  the subspace spanned by $u_1,\ldots,u_K$. Defining  $\Pi_{\spann (\U)}v$ as the projection of $v \in \LLemp$ onto $\spann (\U)$, and $\Pi_{\spann (\U) \cap V_{ \bar{\bnu}}(\Omega)}v$ as  the projection of $v$ onto the closed convex set $\spann (\U) \cap V_{ \bar{\bnu}}(\Omega)$, then we have

\begin{prop} \label{prop:CPCA} Let $\omega_{i} = \log_{\bar{\bnu}}(\nu_{i})$ for $1 \leq i \leq n$, and 
 $\U^{\ast}=\{ u^{\ast}_1,\ldots,u^{\ast}_k \}$ be a minimizer of
\begin{equation}
\frac{1}{n} \sum_{i=1}^{n}   \|\omega_{i} - \Pi_{\spann (\U) \cap V_{ \bar{\bnu}}(\Omega)} \omega_{i} \|_{\bar{\bnu}}^2, \label{eq:optimCPCA}
\end{equation}
over orthonormal sets $\U=\{ u_1,\ldots,u_K\}$ of functions in $\LLemp$ of dimension $K$ (namely such that $\langle u_{j}, u_{j'} \rangle_{\bar{\bnu}} = 0$ if $j \neq j'$). If we let
$$
G_{\U^{\ast}} := \exp_{\bar{\bnu}}(\spann (\U^{\ast}) \cap V_{\bar{\bnu}}(\Omega)),
$$
then $G_{\U^{\ast}}$ is a principal geodesic subset (PGS) of dimension $k$ of  the  measures $\nu_1,\ldots,\nu_n$, meaning that  $G_{\U^{\ast}}$ belongs to the set of minimizers of the optimization problem \eqref{eq:hatGk}.
\end{prop}


\begin{proof}
For $v\in\LLemp$ and a subset $C\in\LLemp$, we define $d_{\bar{\nu}}(v,C)=\inf_{u\in C}\Vert v-u\Vert_{\bar{\nu}}$. Remark that $\sum_i \omega_i=0$. Hence by Proposition 3.3 in \cite{BGKL13}, if $\U^{\ast}$ minimizes
$$
\frac{1}{n}\sum_{i=1}^nd_{\bar{\nu}}^2(\omega_i,\spann(\U^{\ast})\cap V_{\bar{\nu}}(\Omega))=\frac{1}{n}\sum_{i=1}^n\Vert \omega_i-\Pi_{\spann(\U^{\ast})\cap V_{\bar{\nu}}(\Omega)} \omega_i\Vert^2_{\bar{\bnu}},
$$
then $\spann(\U^{\ast})\cap V_{\bar{\nu}}(\Omega)\in\argmin_{C}\frac{1}{n}\sum_{i=1}^n d_{\bar{\nu}}^2(\omega_i,C)$, where $C$ is taken over all nonempty, closed, convex set of $V_{\bar{\nu}}(\Omega)$ such that $dim(C)\le K$ and $0\in C$. By Proposition 4.3 in \cite{BGKL13}, and since $\log_{\bar{\nu}}(\bar{\nu})=0$, we can conclude that $G^{\ast}$ is a geodesic subset of dimension $K$ which minimizes  \eqref{eq:avg_dist}.
\end{proof}

Thanks to Proposition~\ref{prop:CPCA}, it follows that  GPCA in $\WS$ corresponds to a mapping of the data into the Hilbert space $\LLemp$ which is followed by a PCA in $\LLemp$ that is constrained to lie in the convex and closed subset $V_{\bar{\bnu}}(\Omega)$. This has to be interpreted as a geodesicity constraint coming from the definition of  a PGS in $\WS$.

\section{The log-PCA approach}
\label{sec:linear_pca}

For data in a Riemannian manifold, we recall that log-PCA consists in solving a linearized version of the PGA problem by mapping the whole data set to the tangent space at the Fr\'echet mean through the logarithmic map~\cite{geodesicPCA}. This approach is computationally attractive since it boils down to computing a standard PCA.~\cite{Wang2013} used this idea to linearize geodesic PCA in the Wasserstein space $W_2(\mathbb{R}^d)$, by defining the logarithmic map of a probability measure as the barycentric projection of the transport plan with respect to a template measure. This approach has the two drawbacks (1) and (2) of log-PCA mentioned in Section~\ref{sec:related}. A third limitation inherent to the Wasserstein space is that when this template or mean probability measure is discrete, the logarithmic map cannot be defined anywhere. This is why the authors of~\cite{Wang2013} had to apply the barycentric projection, which is a lossy process, to the data set before being able to map it to the tangent space.

We consider as usual a subset $\Omega \subset \R$. In this setting, $\WS$  is a flat space as shown by the isometry property (P1) of Proposition~\ref{prop_vmur}. Moreover, if the Wasserstein barycenter $\bar{\bnu}$ is assumed to be absolutely continuous, then Definition~\ref{def:explog} shows that the logarithmic map at $\bar{\bnu}$ is well defined everywhere. Under such an assumption, log-PCA in $\WS$ corresponds to the following steps:
\begin{enumerate}
	\item compute the log maps (see Definition~\ref{def:explog}) $\omega_i = \log_{\bar{\bnu}}(\nu_i)$, $i=1,\ldots,n$,
	\item perform the PCA of the projected data $\omega_1,\cdots,\omega_n$ in the Hilbert space $\LLemp$ to obtain $K$ orthogonal directions $\tilde{u}_1,\ldots,\tilde{u}_{K}$ in $\LLemp$ of principal variations,
	\item recover a principal subspace of variation in $\WS$ with  the exponential map  $\exp_{\bar{\bnu}}(\spann (\tilde{\U}))$ of the principal eigenspace $\spann(\tilde{\U})$ in $\LLemp$  spanned  by $\tilde{u}_1,\ldots,\tilde{u}_{K}$.
\end{enumerate}

For specific datasets, log-PCA in $\WS$ may be equivalent to GPCA, in the sense that the set $\exp_{\bar{\bnu}}(\spann (\tilde{\U}) \cap V_{\bar{\bnu}}(\Omega))$ is  a principal geodesic subset of dimension $K$ of  the  measures $\nu_1,\ldots,\nu_n$, as defined by \eqref{eq:hatGk}. Informally, this case corresponds to the setting where the data are sufficiently concentrated around their Wasserstein barycenter $\bar{\bnu}$ (we refer to Remark 3.5 in~\cite{BGKL13} for further details). However, carrying out a PCA in the tangent space of $W_2(\mathbb{R})$ at $\bar{\bnu}$ is a relaxation of the convex-constrained GPCA problem \eqref{eq:optimCPCA}, where the elements of the sought principal subspace do not need to be in  $V_{ \bar{\bnu}}$. Indeed, standard PCA in the Hilbert space $\LLemp$, amounts to find $\tilde{\U} = \{ \tilde{u}_1,\ldots,\tilde{u}_{K}\}$ minimizing,
\begin{equation}
\frac{1}{n} \sum_{i=1}^{n}   \|\omega_{i} - \Pi_{\spann (\U)} \omega_{i} \|_{\bar{\bnu}}^2, \label{eq:linearPCA},
\end{equation}
over orthonormal sets $\U=\{ u_1,\ldots,u_k\}$ of functions in $\LLemp$.  It is worth noting that the three steps of log-PCA in $\WS$ are simple to implement and fast to compute, but that performing log-PCA or GPCA \eqref{eq:optimCPCA} in $\WS$ is not necessarily equivalent.

Log-PCA is generally used for two main purposes. The first one is to obtain a low dimensional representation of each data measure $\nu_i =\exp_{\bar{\bnu}}(  \omega_{i}) $ through the coefficients $\dotprod{\omega_i}{\tilde{u}_k}_{L^2_{\bar{\bnu}}}$. 
From this low dimensional representation, the measure  $\nu_i  \in\WS$ can be approximated through the exponential mapping $\exp_{\bar{\bnu}}(\Pi_{\spann (\U)} \omega_{i})$. The second one is to visualize each mode of variation in the dataset, by considering  the evolution of the curve $t \mapsto \exp_{\bar{\bnu}}(t\tilde{u}_k)$ for each $\tilde{u}_k \in \tilde{\U}$.

However, relaxing the convex-constrained GPCA problem \eqref{eq:optimCPCA} when using log-PCA results in several issues. Indeed, as shown in the following paragraphs, not taking into account this geodesicity constraint makes difficult the computation and interpretation of $\exp_{\bar{\bnu}}(\spann (\tilde{\U}))$ as a principal subspace of variation, which may limit its use for data analysis. 

\paragraph{Numerical implementation of pushforward operators} A first downside to the log-PCA approach is the difficulty of the numerical implementation of the pushforward operator in the exponential map $\exp_{\bar{\bnu}}(v)=({\rm id}+v) \# \bar{\bnu}$ when the mapping ${\rm id}+v$ is not a strictly increasing function for a given vector $v \in \spann (\tilde{\U})$. This can be shown with the following proposition, which provides a formula for computing the density of a pushforward operator.
\begin{prop} \label{prop:pushforward_density} (Density of the pushforward) Let $\mu \in W_2(\mathbb{R})$ be an absolutely continuous measure with density $\rho$ (that is possibly supported on an interval $\Omega \subset \R$). Let $T: \mathbb{R} \rightarrow \mathbb{R}$ be a differentiable function such that $|T'(x)| > 0$ for almost every $x \in \R$, and define $\nu = T\# \mu$. Then, $\nu$ admits a density $g$ given by,
\begin{equation}
	\label{eq:pushforward_density}
	g(y) = \sum_{x\in T^{-1}(y)} \frac{\rho(x)}{|T'(x)|}, \; y \in \R.
\end{equation}
When $T$ is injective, this simplifies to,
\begin{equation}
	\label{eq:pushforward_density_injective}
	g(y) = \frac{\rho(T^{-1}(y))}{|T'(T^{-1}(y))|}.
\end{equation}
\end{prop} 
\begin{proof}
Under the assumptions made on $T$, the coarea formula (which is a more general form of Fubini's theorem, see e.g. \cite{krantz2008geometric} Corollary 5.2.6 or \cite{evans2015measure} Section 3.4.3) states that, for any measurable function $h : \R \to \R$, one has
\begin{equation}
\int_\mathbb{R} h(x)|T'(x)|dx = \int_\mathbb{R}\sum_{x\in T^{-1}(y)} h(x) dy. \label{eq:coarea}
\end{equation}
Let $B$ a Borel set and choose $h(x) = \frac{\rho(x)}{|T'|}\mathbf{1}_{T^{-1}(B)}, \; x$. Hence, using \eqref{eq:coarea}, one obtains that
\begin{equation*}
	\begin{split}
	\int_{T^{-1}(B)} \rho(x)dx  & = \int_\mathbb{R}\sum_{x\in T^{-1}(y)} \frac{\rho(x)}{|T'(x)|}\mathbf{1}_{T^{-1}(B)}(x) dy = \int_B\sum_{x\in T^{-1}(y)} \frac{\rho(x)}{|T'(x)|} dy.\\
	\end{split}
\end{equation*}
The definition of the pushforward $\nu(B) = \mu(T^{-1}(B))$ then  completes the proof.
\end{proof}
The numerical computation of formula \eqref{eq:pushforward_density} or \eqref{eq:pushforward_density_injective} is not straightforward. When $T$ is not injective, computation of the formula \eqref{eq:pushforward_density} must be done carefully by partitioning the domain of $T$ in sets on which $T$ is injective. Such a partitioning depends on the method of interpolation for estimating a continuous density $\rho$ from a finite set of its values on a grid of reals. More importantly, when $T'(x)$ is very small, $\frac{\rho(x)}{T'(x)}$ may become very irregular and the density of $\nu = T\# \mu$ may  exhibit large peaks, see Figure~\ref{echec_logPCA} for an illustrative example.

\begin{figure}[ht!]
\centering
\includegraphics[width= 0.72\textwidth,height=0.49\textwidth]{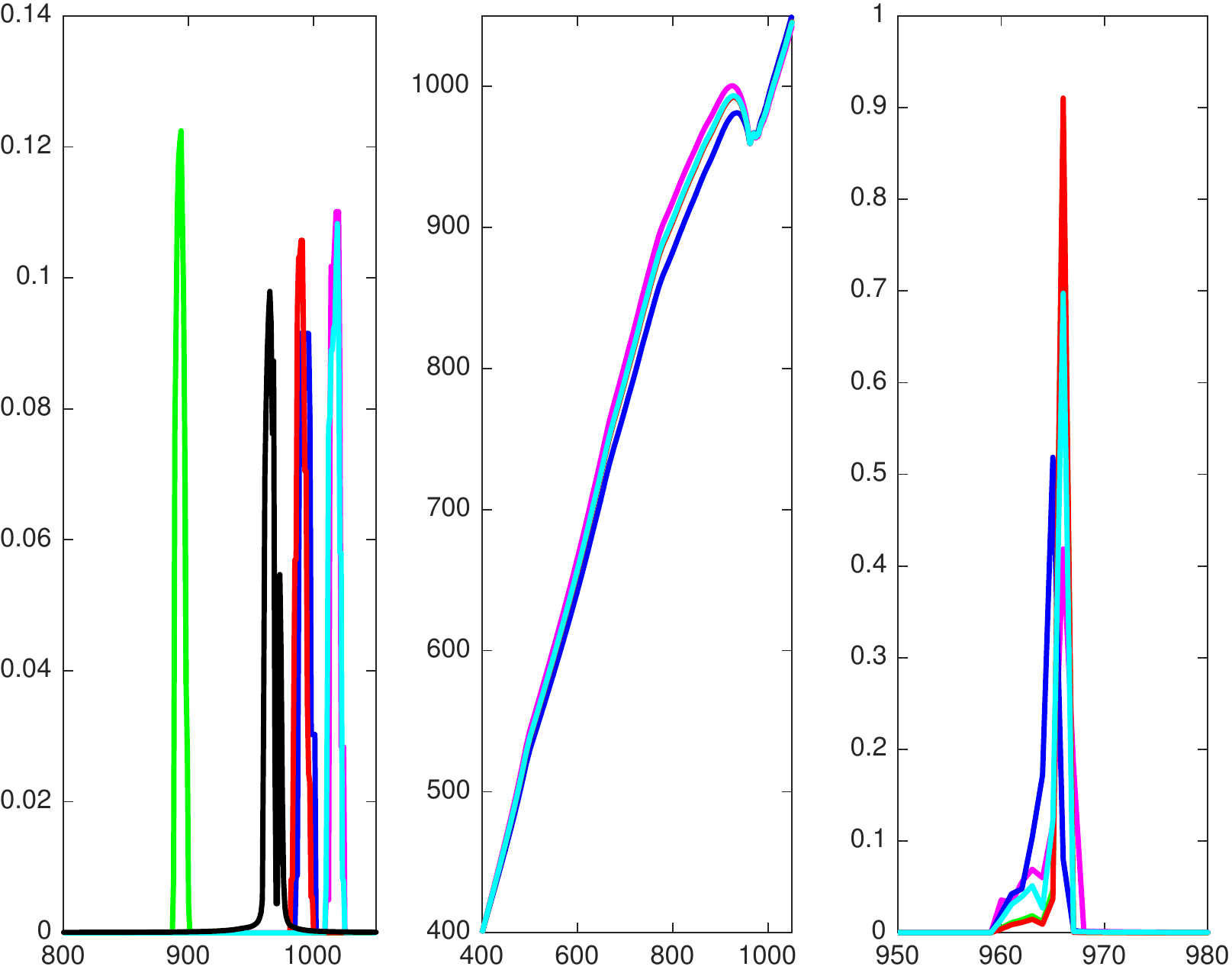} 
\caption{(Left) Distribution of the total precipitation (mm) collected in a year in $1 \leq i \leq 5$ stations among 60 in China - Source : Climate Data Bases of the People's Republic of China 1841-1988 downloaded from http://cdiac.ornl.gov/ndps/tr055.html. The black curve is the density of the Wasserstein barycenter of the 60 stations. (Middle) Mapping $T_{i} = {\rm id} + \Pi_{\spann (\tilde{u}_2)} \omega_{i} $ obtained from the projections of these 5 distributions onto the second eigenvector $\tilde{u}_2$ given by log-PCA of the whole dataset. (Right) Pushforward $\exp_{\bar{\bnu}}(\Pi_{\spann (\tilde{u}_2)} \omega_{i})= T_{i} \# \bar{\bnu}$ of the Wasserstein barycenter $\bar{\bnu}$ for each $1 \leq i \leq 5$. As the derivative $T'_{i}$ take very small values, the densities of the pushforward barycenter $T_{i}\# \bar{\bnu}$  for $1 \leq i \leq 5$ exhibit large peaks (between 0.4 and 0.9) whose amplitude is  beyond the largest values in the original data set (between 0.08 and 0.12).
}\label{echec_logPCA}

\end{figure}

\paragraph{Pushforward of the barycenter outside the support $\Omega$.}  A second downside of log-PCA in $\WS$ is that the range of the mapping $\tilde{T}_{i} = \text{id} + \Pi_{\spann (\tilde{\U})} \omega_{i}$ may be larger than the interval $\Omega$. This implies that the density  of the pushforward of the Wassertain barycenter $\bar{\bnu}$  by this mapping, namely  $  \exp_{\bar{\bnu}}(\Pi_{\spann (\tilde{\U})} \omega_{i})$, may have a support which is not included in $\Omega$. This issue may be critical when trying to estimate the measure $\nu_i =\exp_{\bar{\bnu}}(  \omega_{i}) $ by its projected measure $  \exp_{\bar{\bnu}}(\Pi_{\spann (\tilde{\U})} \omega_{i})$. For example, in a dataset of histograms with bins necessarily containing only positive reals, a projected distribution with positive mass on negative reals would be hard to interpret. 

\paragraph{A higher Wasserstein reconstruction error.} 

Finally, relaxing the geodesicity constraint \eqref{eq:optimCPCA} may actually increase the Wasserstein reconstruction error with respect to  the Wasserstein distance. To state this issue more clearly, we define the  reconstruction error of log-PCA as
\begin{equation}\label{eq:residuals:logPCA}
\tilde{r}(\tilde{\U}) = \frac{1}{n} \sum_{i=1}^{n} d_W^2\big(\nu_i,  \exp_{\bar{\bnu}}(\Pi_{\spann (\tilde{\U})} \omega_{i}) \big).
\end{equation}
and the reconstruction error of GPCA as
\begin{equation}\label{eq:residuals:GPCA}
r(\U^{\ast}) = \frac{1}{n} \sum_{i=1}^{n} d_W^2\big(\nu_i,  \exp_{\bar{\bnu}}(\Pi_{\spann (\U^{\ast}) \cap V_{\bar{\bnu}}(\Omega)} \omega_{i}) \big).
\end{equation}
where $\U^{\ast}$ is a minimiser of \eqref{eq:optimCPCA}. Note that in \eqref{eq:residuals:logPCA}, the projected measures $\exp_{\bar{\bnu}}(\Pi_{\spann (\tilde{\U})} \omega_{i})$ might have a support that lie outside $\Omega$. Hence, the Wasserstein distance $d_W$ in \eqref{eq:residuals:logPCA} has  to be understood for measures supported on $\R$ (with the obvious extension to zero of $\nu_i$ outside $\Omega$).

 The Wasserstein reconstruction error $\tilde{r}(\tilde{\U})$ of log-PCA is the sum of the Wasserstein distances of each data point $\nu_i$ to a point on the surface $\exp_{\bar{\bnu}}(\spann (\tilde{\U}))$ which is given by the decomposition of $\omega_i$ on the orthonormal basis $\tilde{\U}$. However, by Proposition~\ref{prop_vmur}, the isometry property (P1) only holds  between $W_2(\mathbb{R})$ and the convex subset $V_{ \bar{\bnu}} \subset L^2_{\bar{\bnu}}(\mathbb{R})$. Therefore, we may not have $d_W^2\big(\nu_i, \exp_{\bar{\bnu}}\big( \Pi_{\spann (\tilde{\U)}} \omega_{i} \big) \big) =  \|\omega_{i} - \Pi_{\spann (\tilde{\U})} \omega_{i} \|_{\bar{\bnu}}^2$ as $\Pi_{\spann (\tilde{\U})} \omega_{i}$ is a function belonging to $L^2_{\bar{\bnu}}(\mathbb{R})$ which may not necessarily be in $V_{\bar{\bnu}}$. In this case, the minimal Wasserstein distance between $\nu_{i}$ and the surface $\exp_{\bar{\bnu}}(\spann (\U^{\ast}))$ is not equal to $\|\omega_{i} - \Pi_{\spann (\U)} \omega_{i} \|_{\bar{\bnu}}$, and this leads to situations where $\tilde{r}(\tilde{\U})  > r(\U^{\ast})$ as illustrated in Figure~\ref{residual}.
\begin{figure}[ht!]
\centering
\includegraphics[width=0.4\textwidth,height=0.36\textwidth]{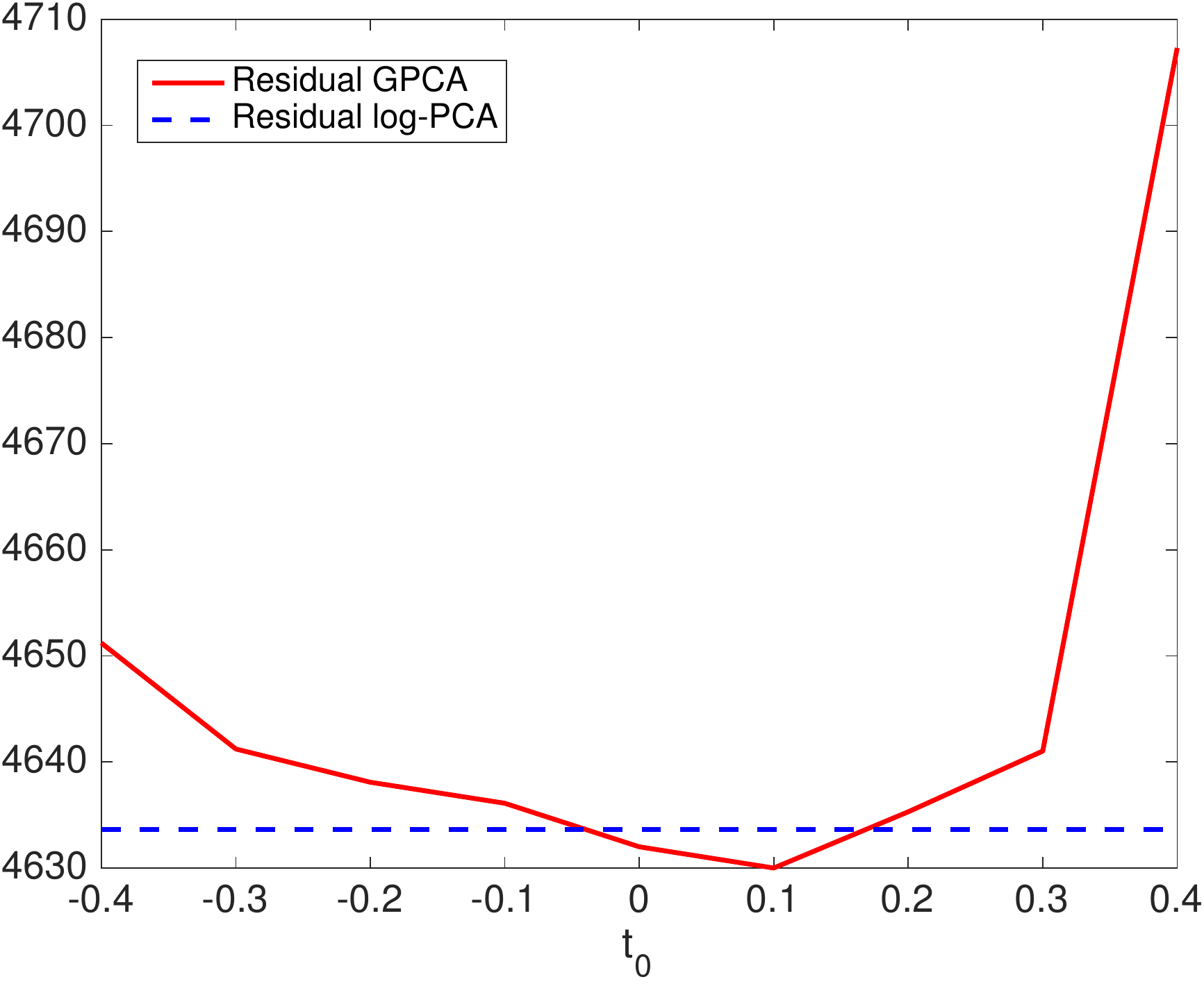}\vspace{-0.2cm}
\caption{Comparison of the Wasserstein reconstruction error between GPCA and log-PCA on the synthetic dataset displayed in Figure~\ref{eucli_wass_gaussian} for the first component, with an illustration of the role of the parameter $t_{0}$ in \eqref{eq:optimCPCA3}.}
\label{residual}
\end{figure}

\section{Two algorithmic approaches for GPCA in $\WS$, for $\Omega\subset\R$} \label{sec:reformulation}

In this section, we introduce two algorithms which solve some of the issues of log-PCA that have been raised in Section \ref{sec:linear_pca}. First, the output of the proposed algorithms guarantees that the computation of mappings to pushforward the Wassertein barycenter to approximate elements in the datasets are strictly increasing (that is they are optimal). As a consequence, the resulting pushforward density behaves numerically much better. Secondly, the geodesic curve or surface are constrained to lie in $\WS$, implying that the projection of the data are distributions whose supports do not lie outside $\Omega$. 

\subsection{Iterative geodesic approach}

In this section, we propose an algorithm to solve a variant of the convex-constrained GPCA problem \eqref{eq:optimCPCA}. Rather than looking for a geodesic subset of a given dimension which fits well the data, we find iteratively orthogonal principal geodesics (i.e. geodesic set of dimension one).
Assuming that that we already know  a subset $\U^{k-1} \subset \LLemp$ containing $k-1$ orthogonal principal directions $\{u_l\}_{l=1}^{k-1}$ (with $\U^{0} = \emptyset$), our goal is  to find a new direction ${u}^{}_k\in\LLemp$ of principal variation by solving the optimisation problem:
\begin{equation}
{u}^{}_k\in\uargmin{v\perp{\U}^{k-1}}\frac{1}{n} \sum_{i=1}^{n}   \|\omega_{i} - \Pi_{\spann({v}) \cap V_{ \bar{\bnu}}(\Omega)} \omega_{i} \|_{\bar{\bnu}}^2, \label{eq:optimCPCA2}
\end{equation}

where the infimum above is taken over all $v \in \LLemp$ belonging to the orthogonal of $\U^{k-1}$. This iterative process is not equivalent to the GPCA problem \eqref{eq:optimCPCA}, with the exception of the first principal geodesic ($k=1$). Nevertheless, it computes principal subsets $\U^{k}$ of dimension $k$ such that the projections of the data onto every direction of principal variation lie in the convex set $V_{ \bar{\bnu}}$.
%

The following proposition is the key result to derive an algorithm to solve \eqref{eq:optimCPCA2} on real data. 
\begin{prop} \label{prop:equivGPCA}
Introducing the characteristic function of the convex set $V_{\bar{\bnu}}(\Omega)$ as:
$$\chi^{}_{V_{\bar{\bnu}}(\Omega)}(v)=\left\{\begin{array}{ll}0&if\,  v\in V_{\bar{\bnu}}(\Omega)\\
+\infty&otherwise\end{array}\right.
$$
the optimisation problem \eqref{eq:optimCPCA2} is equivalent to 
\begin{equation}
{u}^{}_k=\uargmin{v\perp{\U}^{k-1}}
\min_{t_0\in[-1;1]} H(t_{0},v),  \label{eq:optimCPCA3}
\end{equation}
where
\begin{equation}
\label{eq:optimCPCA3_1D}
H(t_{0},v) = \frac{1}{n} \sum_{i=1}^{n}  \min_{t_i\in[-1;1]} \|\omega_{i} - (t_0+t_i) v \|^2_{\bar{\bnu}}+\chi^{}_{V_{\bar{\bnu}}(\Omega)}((t_0-1)v)+\chi^{}_{V_{\bar{\bnu}}(\Omega)}((t_0+1)v). 
\end{equation}
\end{prop}
\begin{proof}
We first observe that  $\Pi_{\spann({u}) \cap V_{ \bar{\bnu}}(\Omega)} \omega_{i}=\beta_i u$, with  $\beta_i\in\mathbb{R}$ and $\beta_i u\in   V_{ \bar{\bnu}}(\Omega)$.
Hence, for $u_k$ solution of \eqref{eq:optimCPCA2}, we have:
$$\frac{1}{n} \sum_{i=1}^{n}   \|\omega_{i} - \Pi_{\spann({u^{}_k}) \cap V_{ \bar{\bnu}}(\Omega)} \omega_{i} \|_{\bar{\bnu}}^2=\frac{1}{n} \sum_{i=1}^{n}   \|\omega_{i} -\beta_i u^{}_k  \|_{\bar{\bnu}}^2.$$
such that $\beta_i\in\mathbb{R}$ and $\beta_i u^{}_k\in   V_{ \bar{\bnu}}(\Omega)$ for all $i\in\{1,\ldots,n\}$. We take $M\in\argmax_{1 \leq i \leq n}\beta_i$ and $m\in\argmin_{1 \leq i \leq n}\beta_i$. Without loss of generality, we can assume that $\beta_M> 0$ and $\beta_m< 0$. We then define $v=(\beta_M-\beta_m) u^{}_k/2$ and  $t_0=(\beta_M+\beta_m)/(\beta_M-\beta_m)$, that checks $|t_0|<1$. Hence, for all $i=1,\ldots, n$, there exists $t_i\in[-1;1]$ such that:
$\beta_i u^{}_k=(t_0+t_i)v\in   V_{ \bar{\bnu}}$.
In particular, one has $t_M=1$ and $t_m=-1$, which means that $(t_0\pm 1)v\in V_{\bar{\bnu}}(\Omega)$. Reciprocally, $(t_0\pm 1)v\in V_{\bar{\bnu}}(\Omega)$ insure us by convexity of $V_{\bar{\bnu}}(\Omega)$ that for all $t_i\in[-1;1]$, $(t_0+t_i)v\in V_{\bar{\bnu}}(\Omega)$.
\end{proof}

Proposition \ref{prop:equivGPCA} may be interpreted as follows. For a given $t_{0} \in [-1;1]$, let $v \in \perp{\U}^{k-1}$ satisfying $(t_0-1)v \in V_{\bar{\bnu}}$ and $(t_0+1)v \in V_{\bar{\bnu}}$. Then, if one defines the curve
\begin{equation}
\label{eq:geodesic_parameterization}
g_t(t_0, v) = (id + (t_0 + t)v)\# \bar{\bnu} \mbox{ for } t \in [-1;1],
\end{equation}
it follows, from Lemma \ref{lem:gamma},  that $(g_t(t_0,v))_{t \in [-1;1]}$  is a geodesic since it can be written as $g_t(t_0,v) = \exp_{\bar{\bnu}}((1-u)w_0+uw_1), u \in [0,1]$ with $w_0 = (t_0-1)v,~w_1 = (t_0+1)v$, $u = (t+1)/2$, and with $w_0$ and $w_1$ belonging to $V_{\bar{\bnu}}$ for $|t_0|<1$. From the isometry property (P1) in Proposition \ref{prop_vmur},  one has
\begin{equation}
\label{eq:iterative_geodesic_isometry}
\min_{t_i\in[-1;1]} \|\omega_{i} - (t_0+t_i) v \|^2_{\bar{\bnu}} = \min_{t_i\in[-1;1]} d_W^2(\nu_i, g_{t_i}(v)),
\end{equation}
and thus the objective function $H(t_{0},v)$ in \eqref{eq:optimCPCA3} is equal to the sum of the squared Wasserstein  distances between the dataset to the geodesic curve  $(g_t(t_0,v))_{t \in [-1;1]}$.

The choice of the parameter $t_{0}$ corresponds to the location of the mid-point of the geodesic $g_t(t_0,v)$, and it plays a crucial role. Indeed, the minimisation of  $H(t_{0},v)$ over $t_{0} \in [-1;1]$ in \eqref{eq:optimCPCA3} cannot be avoided to obtain an optimal Wasserstein reconstruction error. This is illustrated by the Figure \ref{residual}, where the  Wasserstein reconstruction error $\tilde{r}(\tilde{\U})$ of log-PCA (see relation \eqref{eq:residuals:logPCA}) is compared with the ones of GPCA, for different $t_0$, obtained for $k=1$ as\vspace{-0.1cm}
\begin{equation*}
t_{0} \in [-1;1] \mapsto H(t_{0},{u}^{t_{0}}_1)
\end{equation*}
with ${u}^{t_{0}}_1=\argmin_{v} H(t_{0},v)$.We therefore exhibit that GPCA can lead to a better low dimensional  data representation than log-PCA in term of Wasserstein residuals.

\subsection{Geodesic surface approach}

Once a family of vectors $(v_1,\cdots,v_k)$ has been found through the minisation of problem \eqref{eq:optimCPCA2}, one can recover a geodesic subset of dimension $k$ by considering all convex combinations of the vectors $((t_0^1+1)v_1,(t_0^1-1)v_1,\cdots,(t_0^k+1)v_k,(t_0^k-1)v_k)$. However, this subset may not be a solution of \eqref{eq:optimCPCA} since we have no guarantee that a data point $\nu_i$ is actually close to this geodesic subset. This discussion suggests that we may consider solving the GPCA problem \eqref{eq:optimCPCA} over geodesic set parameterized as in Proposition \ref{eq:optimCPCA2}. 
In order to find principal geodesic subsets which are close to the dataset, we consider a family $V^K=(v_1,\cdots,v_K)$ of linearly independant vectors and $\mathbf{t}_0^K = (t_0^1,\cdots,t_0^K) \in [-1,1]^K$ such that $(t_0^1-1)v_1,(t_0^1+1)v_1,\cdots,(t_0^K-1)v_K,(t_0^K+1)v_K$ are all in $V_{ \bar{\bnu}}$. Convex combinations of the latter family  provide a parameterization of a geodesic set of dimension $K$ by taking the exponential map $\exp_{\bar{\bnu}}$ of\vspace{-0.1cm}
\begin{equation}
	\hat{V}_{\bar{\bnu}}(V^K,\mathbf{t}_0^K) = \{ \sum_{k=1}^K(\alpha_{k}^{+}( t_0^k+1)+\alpha_{k}^{-}(t_0^k-1))v_k, \alpha^{\pm}\in A \}\vspace{-0.2cm}
\end{equation}
where $A$ is a simplex constraint: $\alpha^{\pm}\in A \Leftrightarrow  \alpha_{k}^{+},\alpha_{k}^{-}\geq 0$ and $\sum_{k=1}^K(\alpha_{k}^{+}+\alpha_{k}^{-})\leq 1$. We hence substitute the general sets $\spann (\U) \cap V_{ \bar{\bnu}}(\Omega)$ in the definition of the GPCA problem \eqref{eq:optimCPCA} to obtain,\vspace{-0.1cm}
\begin{eqnarray} \label{eq:optimALL2}
(u_1,\cdots,u_K) &= \uargmin{V^K, \mathbf{t}_0^K}&\frac{1}{n} \sum_{i=1}^{n}   \|\omega_{i} - \Pi_{\hat{V}_{\bar{\bnu}}(V^K,t^K)} \omega_{i} \|_{\bar{\bnu}}^2,\nonumber\\
&=\uargmin{v_1,\cdots,v_K}&\umin{\mathbf{t}_0^K\in [-1,1]^K} \frac{1}{n}\sum_{i=1}^n\umin{\alpha_{i}^{\pm} \in A}\Vert \omega_i-\sum_{k=1}^K(\alpha_{ik}^{+}( t_0^k+1)+\alpha_{ik}^{-}(t_0^k-1))v_k\Vert_{\bar{\bnu}}^2\\
&& \hspace{0.6cm}+\sum_{k=1}^K
\left(\chi_{V_{\bar{\bnu}}(\Omega)}((t_0^k+1)v_k)
+\chi_{V_{\bar{\bnu}}(\Omega)}((t_0^k-1)v_k)\right)+\sum_{i=1}^n\chi_A(\alpha_{i}^{\pm}).\nonumber
\end{eqnarray}


\subsection{Discretization and Optimization}\label{sec:algo}

In this section we follow the framework of the iterative geodesic algorithm. We provide additional details When the optimization procedure of the geodesic surface approach differs from the iterative one.

\subsubsection{Discrete optimization problem}

Let $\Omega = [a;b]$ be a compact interval, and consider its discretization over $N$ points $a=x_1<x_2<\cdots <x_N=b$,  $\Delta_j=x_{j+1}-x_j$, $j=1,\ldots, N-1$.
We recall that the functions $\omega_{i} = \log_{\bar{\bnu}}(\nu_{i})$ for $1 \leq i \leq n$ are elements of  $\LLemp$ which correspond to the mapping of the data to the tangent space at the Wasserstein barycenter $\bar{\bnu}$.  In what follows, for each $1 \leq i \leq n$, the discretization of the function $\omega_{i}$  over the grid reads $\mathbf{w}_{i}=(w_{i}^j)_{j=1}^N\in \R^N$. 
We also recall that $\chi^{}_A(u)$ is the characteristic function of a given set $A$, namely $\chi^{}_A(u)=0$ if $u\in A$ and $+\infty$ otherwise. Finally, the space $\R^N$ is understood to be endowed with the following inner product and norm 
$
\langle \mathbf{u},\mathbf{v} \rangle_{\bar{\bnu}} = \sum_{j=1}^N \bar{\bfun}(x_j) u_j v_j$  and $\| \mathbf{v} \|^{2}_{\bar{\bnu}}=\langle  \mathbf{v},  \mathbf{v} \rangle_{\bar{\bnu}} 
$ for $\mathbf{u}, \mathbf{v} \in \R^N, $ 
where $\bar{\bfun}$ denotes the density of the measure $\bar{\bnu}$.
Let us now suppose that we have already computed  $k-1$ orthogonal (in the sense $\langle \mathbf{u},\mathbf{v} \rangle_{\bar{\bnu}}=0$) vectors $\mathbf{u}^{}_1,\cdots \mathbf{u}^{}_{k-1}$ in $\R^N$ which stand for the discretization of orthonormal functions $u_1,\ldots,u_{k-1}$ in $\LLemp$ over the grid $(x_j)_{j=1}^N$.

Discretizing problem \eqref{eq:optimCPCA3} for a fixed $t_0\in]-1;1[$, our goal is  to find a new direction $\mathbf{u}^{}_k\in\R^N$ of principal variations by solving the following problem over all $\mathbf{v}=\{v_j\}_{j=1}^N\in\R^N$: \vspace{-0.2cm}
\begin{equation}
\mathbf{u}^{}_k\in\uargmin{\mathbf{v}\in\R^N}  \frac{1}{n} \sum_{i=1}^n \left(  \min_{t_i\in[-1;1]} \| \mathbf{w}_i-(t_0+t_i)\mathbf{v} \|^{2}_{\bar{\bnu}}\right) +\chi^{}_S(\mathbf{v}) +\chi^{}_{V}((t_0-1)\mathbf{\mathbf{v}})+\chi^{}_{V}((t_0+1)\mathbf{\mathbf{v}})
,\vspace{-0.1cm} \label{eq:discretepb}
\end{equation}
where 
$S=\{\mathbf{v}\in\R^N\, \textrm{s.t. } \langle \mathbf{v},\mathbf{u}^{}_l \rangle_{\bar{\bnu}} =0,\,  l=1\cdots k-1\},$
is a convex set that deals with the orthogonality constraint $\mathbf{v}\perp\U^{k-1}$ and $V$ corresponds to the discretization of the constraints contained in  $V_{\bar{\bnu}}(\Omega)$. From Proposition \ref{prop_vmur} (P3), we have that $\forall v\in V_{\bar{\bnu}}(\Omega)$, $T:=id+v$ is non decreasing and $T(x)\in\Omega$ for all $x\in\Omega$.
 Hence the discrete convex set $V$ is defined as\vspace{-0.1cm}
$$V=\{\mathbf{v}\in\R^N\, \textrm{s.t. }   x_{j+1}+v_{j+1}\geq x_j+v_j,\,  j=1\cdots N-1 \textrm{ and } x_j+v_j\in [a;b],\, j=1\cdots N\}$$
and  can be rewritten as the intersection of two  convex sets dealing with each constraint separately. 

\begin{prop}
One has\vspace{-0.1cm}
$$\chi^{}_{V}((t_0-1)\mathbf{\mathbf{v}})+\chi^{}_{V}((t_0+1)\mathbf{\mathbf{v}})=\chi^{}_{D}(\mathbf{\mathbf{v}})+\chi^{}_{\spaceK}(K\mathbf{\mathbf{v}}),$$
where the convex sets $D$ and $\spaceK$ respectively deal with the domain constraints $x_j+(t_0+1)v_j\in[a;b]$ and $x_j+(t_0-1)v_j\in[a;b]$, i.e.:\vspace{-0.1cm}
\begin{equation}
\label{op:D}
 D=\{\mathbf{v}\in \R^N,\, \textrm{s.t. } m_j\leq v_j\leq M_j\},
 \end{equation}
with 
$
m_j=\max\left(\frac{a-x_j}{t_0+1},\frac{b-x_j}{t_0-1}\right)$ and $M_j= \min\left(\frac{a-x_j}{t_0-1},\frac{b-x_j}{t_0+1}\right), 
$
and the non decreasing constraint of $id+(t_0\pm1)\mathbf{v}$:\vspace{-0.1cm}
\begin{equation}\label{op:E}\spaceK=\{ \mathbf{z}\in\R^N\, \textrm{s.t. }  -1/(t_0+1)\leq z_j\leq 1/(1-t_0)
\}.\end{equation}
with the differential operator $K:\R^N\to\R^N$  computing the discrete derivative of $\mathbf{v} \in \R^n$ as\vspace{-0.1cm}
\begin{equation}\label{op:K}(K\mathbf{v})_j =\left\{\begin{array}{ll}(v_{j+1}-v_j)/(x_{j+1}-x_j)&\textrm{if } 1\leq j< N\\ 0&\textrm{if }j=N,\end{array}\right.\end{equation} 
%
%
\end{prop}
Having $D$ and $\spaceK$ both  depending  on $t_0$  is not an issue since  problem \eqref{eq:discretepb} is solved for fixed $t_0$.
Introducing $\mathbf{t}=\{t_i\}_{i=1}^n\in\R^n$,  problem \eqref{eq:discretepb} can be reformulated as:
\begin{equation}
\label{eq:pb0}\min_{\mathbf{v}\in\R^N}  \min_{\mathbf{t}\in\R^n}J(\mathbf{v},\mathbf{t}):=\underbrace{\sum_{i=1}^n  \| \mathbf{w}_i-(t_0+t_i)\mathbf{v} \|^{2}_{\bar{\bnu}}}_{F(\mathbf{v},\mathbf{t})}+\underbrace{\chi^{}_S(\mathbf{v})  +\chi^{}_{D}(\mathbf{\mathbf{v}})+\chi^{}_{\spaceK}(K\mathbf{v})+ \chi^{}_{B^n_1}(\mathbf{t})}_{G(\mathbf{v},\mathbf{t})}.
\end{equation}
where $B^n_1$ is the $L^\infty$ ball of $\R^n$ with radius $1$ dealing with the constraint $t_i\in[-1;1]$. Notice that $F$ is differentiable but non-convex  in $(\mathbf{v},\mathbf{t})$ and $G$ is non-smooth and convex.
\paragraph{Geodesic surface approach}
For fixed $(t_0^1,\ldots,t_0^K)\in\R^K$ and $\mathbf{\alpha_{}^{\pm}}=\{\alpha_{k}^{+},\alpha_{k}^{-}\}_{k=1}^K$, the discretized version of \eqref{eq:optimALL2} is then
\begin{align}
\min_{\mathbf{v_1},\ldots,\mathbf{v_K}\in\R^N}  \min_{\mathbf{\alpha_{1}^{\pm}},\ldots,\mathbf{\alpha_{n}^{\pm}}\in\R^{2K}} & F'(\mathbf{v},\mathbf{t})+  G'(\mathbf{v},\mathbf{t}) \label{eq:surf0},
\end{align}
where  $F'(\mathbf{v},\mathbf{t})=\sum_{i=1}^n  \| \mathbf{w}_i- \sum_{k=1}^K  (\alpha_{ik}^{+}( t_0^k+1)  +\alpha_{ik}^{-}(t_0^k-1))\mathbf{v_k}\Vert_{\bar{\bnu}}^2$ is still non-convex and differentiable,  $ G'(\mathbf{v},\mathbf{t}) =\sum_{k=1}^K \left(\chi^{}_{\spaceK}(K\mathbf{v_k})+\chi^{}_{D_k}(\mathbf{\mathbf{v_k}})\right)+ \sum_{i=1}^n \chi^{}_A(\mathbf{\alpha_{i}^{\pm}})^2$ is convex and non smooth, 
 $A$ is the simplex of $\R^{2K}$ and $D_k$ is defined as in \eqref{op:D}, depending on $t_0^k$.
 We recall that  the orthogonality between vectors $\mathbf{v_k}$ is not taken into account  in the geodesic surface approach.

\subsubsection{Optimization through the Forward-Backward Algorithm}
Following \cite{Attouch},  in order to compute a critical point of problem \eqref{eq:pb0}, one can consider the Forward-Backward algorithm (see also \cite{Ochs} for an acceleration using inertial terms). Denoting as $X=(\mathbf{v},\mathbf{t})\in\R^{N+n}$,  taking $\tau>0$ and $X^{(0)}\in\R^{N+n}$, it reads:
\begin{equation}\label{algoFB}X^{(\ell+1)}=\Prox_{\tau G}(X^{(\ell)}-\tau \nabla F(X^{(\ell)})),\end{equation}
where 
$\Prox_{\tau G}(\tilde X)=\argmin_{X} \frac1{2\tau}||X-\tilde X||^2+G(X)$ with the Euclidean norm $||\cdot ||$.
In order to guarantee the convergence of this algorithm,  the gradient of $F$ has to be Lipschitz continuous with parameter $M>0$  and the time step  should be taken as $\tau<1/M$.
The details of computation of $\nabla F$ and $\Prox_{\tau G}$ for the two algorithms are given in Appendix \ref{sec:annexe}.


\section{Statistical comparison between log-PCA and GPCA on  synthetic and real data} \label{sec:num}

\subsection{Synthetic example - Iterative versus geodesic surface approaches}

First, for the synthetic example displayed in Figure \ref{eucli_wass_gaussian}, we compare the two algorithms (iterative and geodesic surface approaches) described in Section \ref{sec:reformulation}. The results are reported in Figure \ref{gpca_gaussian_2_algo} by comparing the projection of the data onto the first and second geodesics computed with each approach. We also also display the projection of the data onto the two-dimensional surface generated by each method. It should be recalled that the principal surface for the iterative geodesic algorithm is not necessarily a geodesic surface but each $g_t(t_0^k,u_k)_{t\in[-1;1]}$ defined by \eqref{eq:geodesic_parameterization} for $k=1,2$ is a geodesic curve for $\U=\{u_1,u_2\}$. For data generated from a location-scale family of Gaussian distributions, it appears that each algorithm provides a satisfactory reconstruction of the data set. The main divergence concerns the first and second principal geodesic. Indeed enforcing the orthogonality between components in the iterative approach enables to clearly separate the modes of variation in location and scaling, whereas searching directly a geodesic surface in the second algorithm implies a mixing of these two types of variation.

Note that the barycenter of Gaussian distributions $\mathcal{N}(m_i,\sigma_i^2)$ can be shown to be Gaussian with mean $\sum m_i$ and variance $(\sum \sigma_i)^2$.
\begin{figure}[ht!]
\centering
\includegraphics[width= 0.7\textwidth,height=0.52\textwidth]{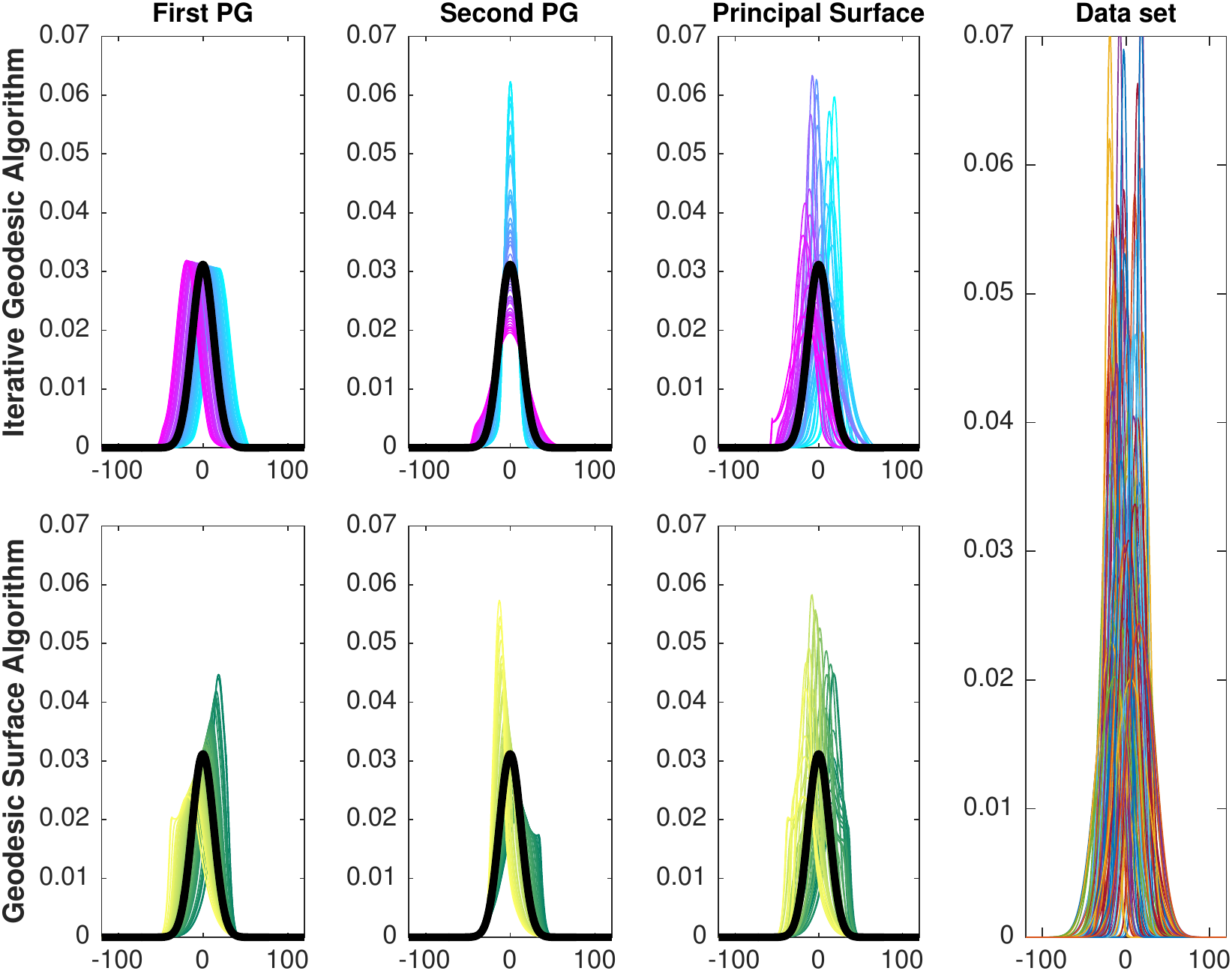} \vspace{-0.2cm}
\caption{Synthetic example - Data sampled from a location-scale family of Gaussian distributions. The first row is the GPCA of the data set obtained with the iterative geodesic approach. The second row is the GPCA through the geodesic surface approach. The black curve is the density of the Wasserstein barycenter. Colors encode the progression of the pdf of principal geodesic components in $\WS$.}
\label{gpca_gaussian_2_algo}
\end{figure}

\subsection{Population pyramids}
As a first real example, we consider a real dataset whose elements are histograms representing the population pyramids of $n=217$ countries for the year 2000 (this dataset is produced by the International Programs Center, US Census Bureau (IPC, 2000), available at \url{http://www.census.gov/ipc/www/idb/region.php}). Each histogram in the database represents the relative frequency by age, of people living in a given country. Each bin in a histogram is an interval of one year,  and the last interval corresponds to people older than 85 years. The histograms are normalized so that their area is equal to one, and thus they represent a set of pdf. In Figure \ref{fig:pop}, we display the population pyramids of 4 countries, and the whole dataset. Along the interval $\Omega = [0,84]$, the variability in this dataset can  be considered as being small.  
\begin{figure}[ht!]
\centering
\includegraphics[width= 0.55\textwidth,height=0.25\textwidth]{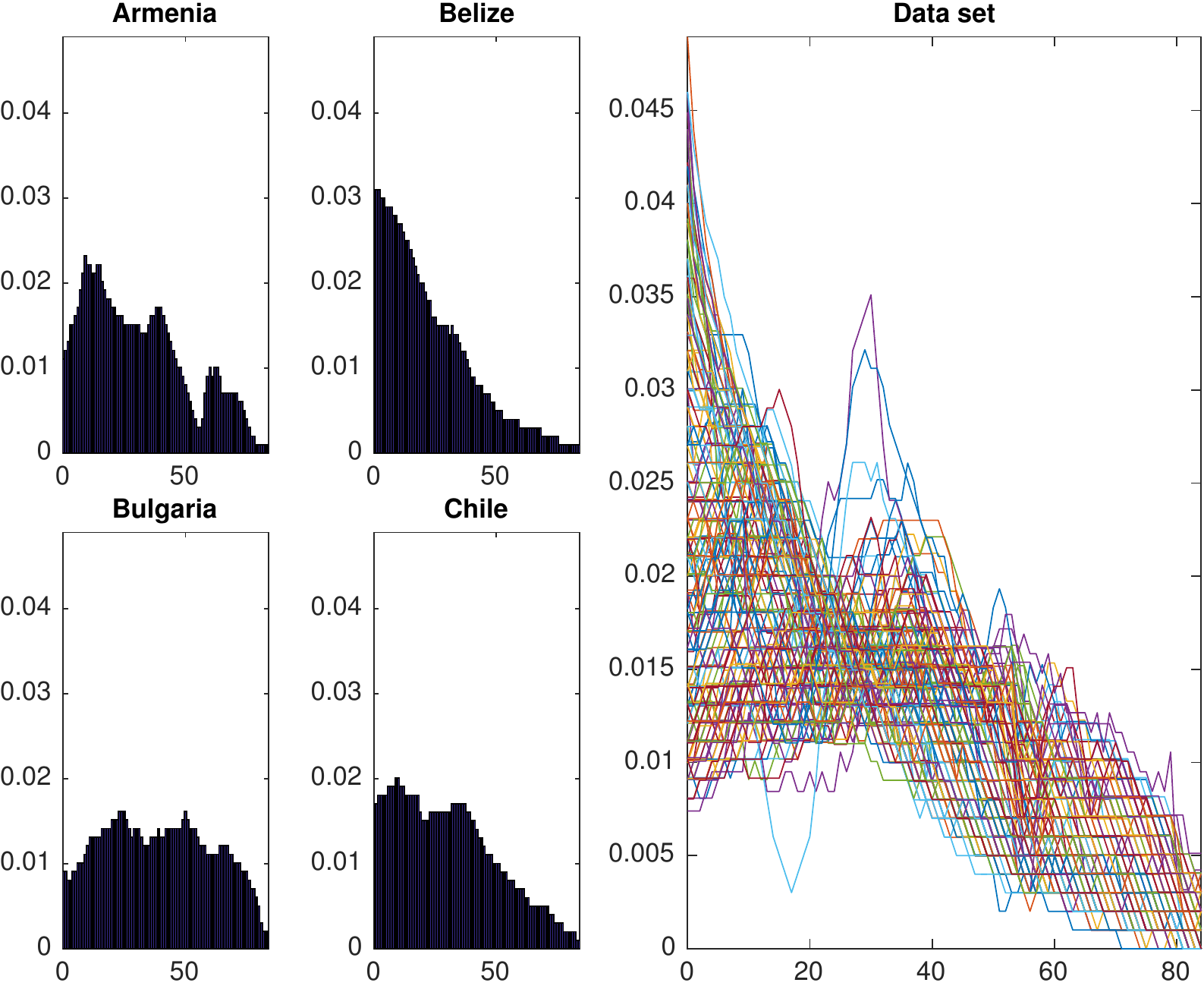} \vspace{-0.2cm}

\caption{\label{fig:pop}Population pyramids. A subset of population pyramids for 4 countries (left) for the year 2000, and the whole dataset of $n=217$ population pyramids (right) displayed as pdf over the interval $[0,84]$.\vspace{-0.2cm}} 
\end{figure}

For $K=2$, log-PCA and the iterative GPCA algorithm lead to the same principal orthogonal directions in $\LLemp$, namely that $\tilde{u}_1 = u^{\ast}_{1}$ and $\tilde{u}_2 = u^{\ast}_{2}$ where $(\tilde{u}_1,\tilde{u}_2)$ minimizes  \eqref{eq:linearPCA} and $(u^{\ast}_{1},u^{\ast}_{2})$ are minimizers of \eqref{eq:optimCPCA3}. In this case, all projections of data $\omega_i = \log_{\bar{\bnu}}(\nu_i)$ for $i=1,\ldots,n$ onto $\spann (\{\tilde{u}_1,\tilde{u}_2  \})$ lie in  $V_{\bar{\bnu}}(\Omega)$, which means that  log-PCA and the iterative geodesic algorithm lead exactly the same principal geodesics. Therefore, population pyramids is an example of data that are sufficiently concentrated around their Wasserstein barycenter so that log-PCA and GPCA are equivalent approaches (see Remark 3.5 in \cite{BGKL13} for further details). 
Hence, we only display in Figure \ref{gpca_pays} the results of the iterative and geodesic surface algorithms.

In the iterative case, the projection onto the first geodesic exhibits the difference between less developed countries (where the population is mostly young) and more developed countries (with an older population structure). The second geodesic captures more subtle divergences concentrated on the middle age population. It can be observed that the geodesic surface algorithm gives different results since the orthogonality constraint on the two principal geodesics is not required. In particular, the principal surface mainly exhibit differences between countries with 
a young population with countries having an older population structure, but the difference between its first and second principal geodesic is less contrasted. 

\begin{figure}[ht!]
\centering
\includegraphics[width= 0.7\textwidth,height=0.52\textwidth]{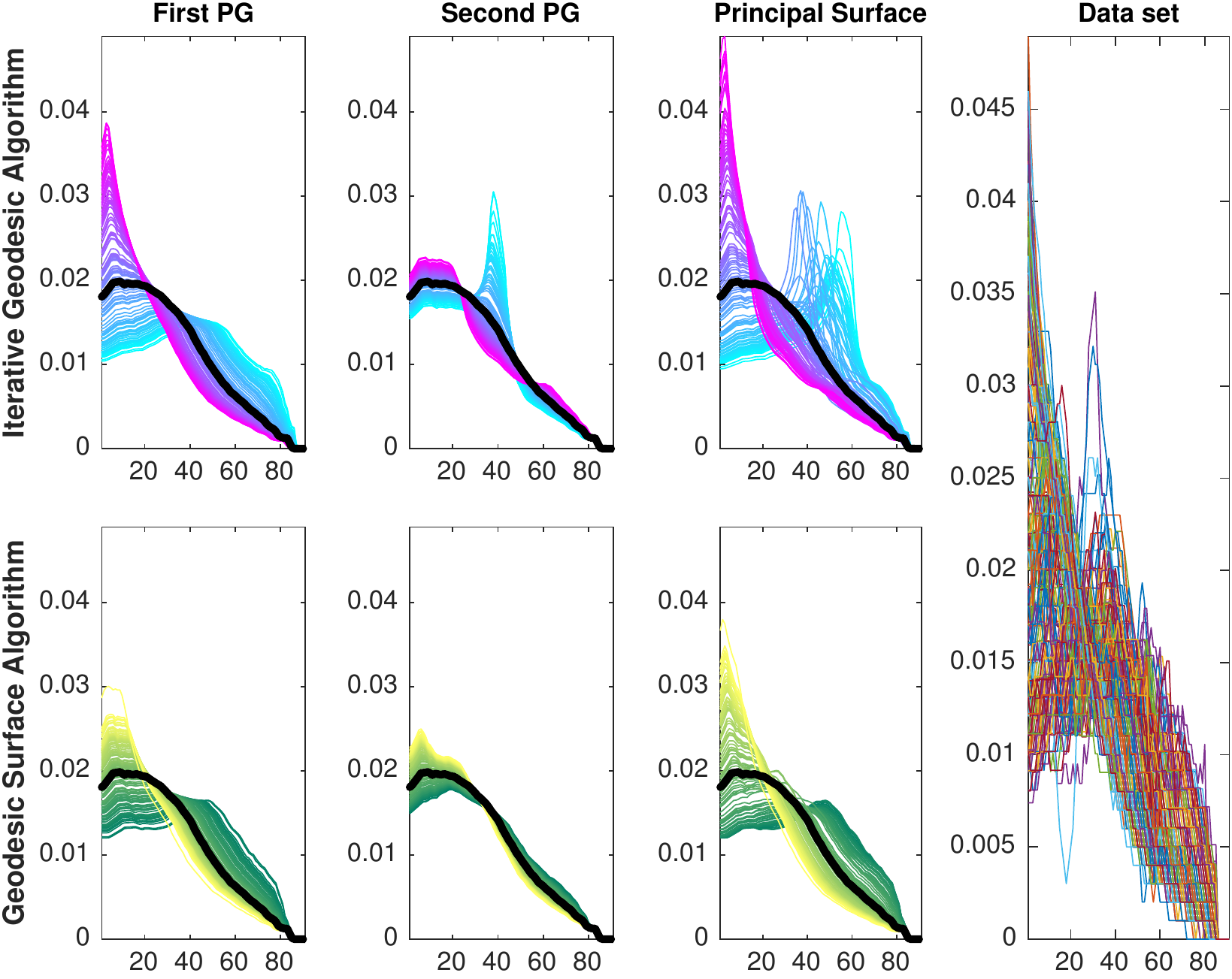} \vspace{-0.2cm}
\caption{\label{gpca_pays}Population pyramids. The first row is the GPCA of the data set obtained with the iterative geodesic approach. The second row is the GPCA through the geodesic surface approach. The first (resp. second) column is the projection of the data into the first (resp. second) principal direction. The black curve is the density of the Wasserstein barycenter. Colors encode the progression of the pdf of principal geodesic components in $\WS$.\vspace{-0.2cm}}
\end{figure}

\subsection{Children's first name at birth}
In a second example, we consider a dataset of  histograms which represent, for a list of $n=1060$ first names, the distribution of children born with that name per year in France between years 1900 and 2013. In Figure \ref{fig:names}, we display the histograms of four different names, as well as the whole  dataset. Along the interval $\Omega = [1900,2013]$, the variability in this dataset is much larger than the one observed for  population pyramids.  This dataset has been provided by the INSEE (French Institute of Statistics and Economic Studies).

\begin{figure}[h!]
\centering

\includegraphics[width= 0.6\textwidth,height=0.25\textwidth]{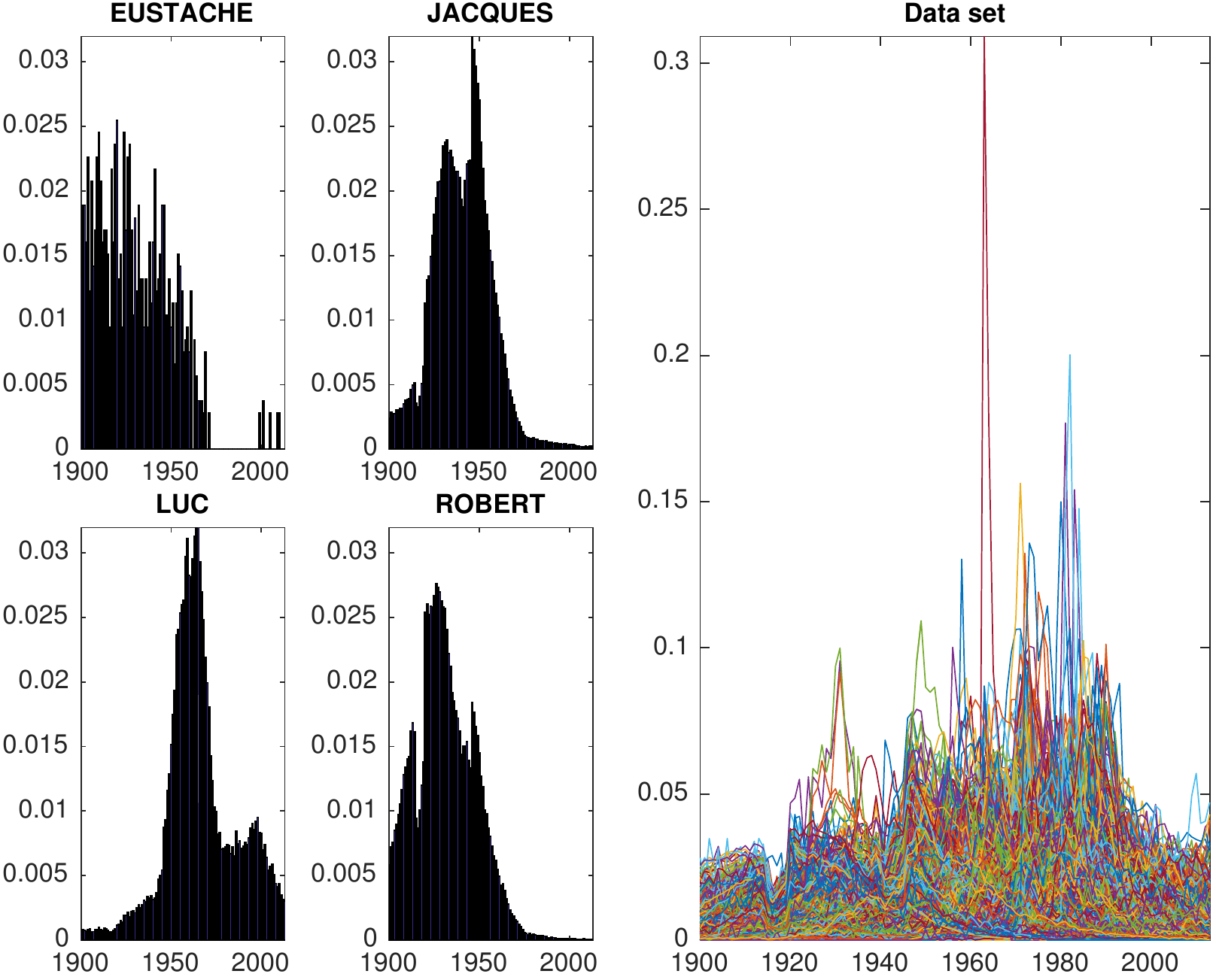} 

\caption{Children's first name at birth. An subet of 4 histograms representing the distribution of children born with that name per year in France, and the whole dataset of $n=1060$ histograms (right), displayed as pdf over the interval $[1900,2013]$} \label{fig:names}
\end{figure}

This is an  example of real data where log-PCA and GPCA are not equivalent procedures for $K=2$ principal components. We recall that log-PCA leads to the computation of  principal orthogonal directions  $\tilde{u}_1,\tilde{u}_2$ in $\LLemp$ minimizing  \eqref{eq:linearPCA}. First observe that in the left column of Figure \ref{maps_name},  for some data $\omega_i = \log_{\bar{\bnu}}(\nu_i)$, the mappings $\tilde{T}_{i} = \text{id} + \Pi_{\spann (\{\tilde{u}_1\})} \omega_{i}$ are decreasing, and their range is larger than the interval $\Omega$ (that is, for some $x \in \Omega$, one has that $\tilde{T}_{i}(x) \notin \Omega$). Hence, such $\tilde{T}_{i}$ are not optimal mappings. Therefore, the condition  $\Pi_{\spann (\tilde{U})} \omega_{i} \in V_{ \bar{\bnu}}(\Omega)$ for all $1 \leq i \leq n$ (with $\tilde{U} = \{\tilde{u}_1,\tilde{u}_2\})$ is not satisfied, implying that log-PCA does not lead to a solution of GPCA thanks to Proposition 3.5 in \cite{BGKL13}.
 \begin{figure}[ht!]
\centering
\includegraphics[scale=0.62]{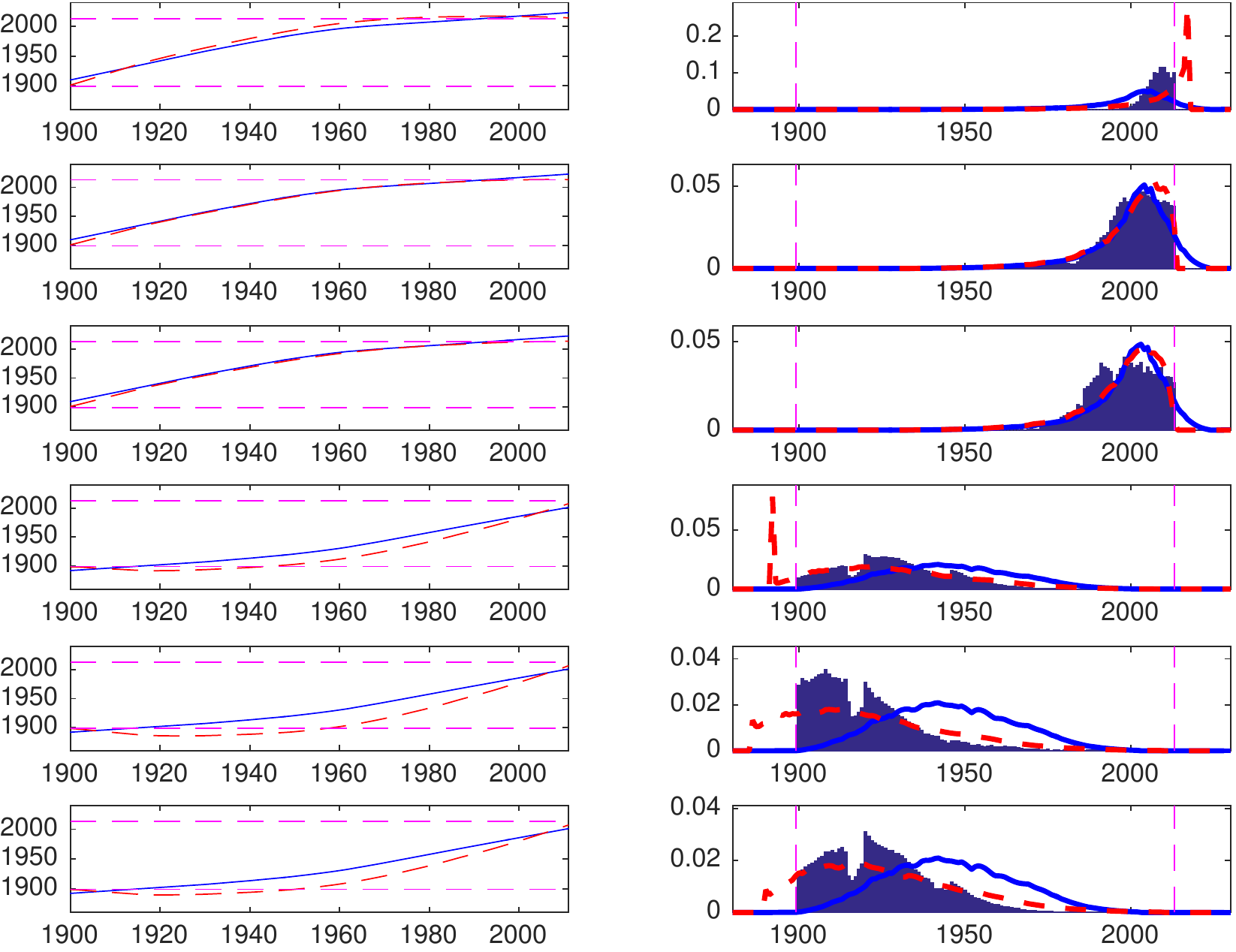} 
\caption{Children's first name at birth with support $\Omega = [1900,2013]$. (Left) The dashed red curves represent the mapping $\tilde{T}_{i} = \text{id} + \Pi_{\spann (\{\tilde{u}_1\})} \omega_{i}$  where $\omega_i = \log_{\bar{\bnu}}(\nu_i)$, and $\tilde{u}_1$ is the first principal direction in $\LLemp$ obtained via log-PCA. The blue curves are the mapping $T_{i} = \text{id} + \Pi_{\spann (\{u^{\ast}_{1}\})} \omega_{i}$, where  $u^{\ast}_{1}$ is the first principal direction in $\LLemp$ obtained via the iterative algorithm. (Right) The histogram stands for the pdf of measures $\nu_i$ that have a large Wasserstein distance with respect to the barycenter $\bar{\bnu}$. The red curves are the pdf of the projection   $\exp_{\bar{\bnu}}(\Pi_{\spann (\tilde{u}_1)} \omega_{i})$ with log-PCA, while the blue curves are the pdf of the projection $\exp_{\bar{\bnu}}(\Pi_{\spann (u^{\ast}_1)} \omega_{i})$ with GPCA.
}
\label{maps_name}
\end{figure}

Hence, for log-PCA, the corresponding histograms displayed in the right column of Figure \ref{maps_name} are  such that $\Pi_{\spann (\{\tilde{u}_1\})} \omega_{i} \notin V_{ \bar{\bnu}}(\Omega)$. This implies that the densities of the projected measures $  \exp_{\bar{\bnu}}(\Pi_{\spann (\tilde{u}_1)} \omega_{i})$ have a support outside $\Omega = [1900,2013]$. Hence, the estimation of  the measure $\nu_i =\exp_{\bar{\bnu}}(  \omega_{i}) $  by its projection onto the first mode of variation obtained with log-PCA is not satisfactory. 

In Figure \ref{maps_name}, we also display the results given by the iterative geodesic algorithm, leading to orthogonal directions  $u^{\ast}_{1},u^{\ast}_{2}$ in $\LLemp$ that are minimizers of \eqref{eq:optimCPCA3}. Contrary to the results obtained with log-PCA, one observes in Figure \ref{maps_name} that all the mapping $T_{i} = \text{id} + \Pi_{\spann (\{u^{\ast}_{1}\})} \omega_{i}$ are non-decreasing, and such that $T_{i}(x) \in \Omega$ for all $x \in \Omega$.
 Nevertheless, by enforcing these two conditions, one has that a good estimation of  the measure $\nu_i =\exp_{\bar{\bnu}}(  \omega_{i}) $ by its projection  $\exp_{\bar{\bnu}}(\Pi_{\spann (u^{\ast}_1)} \omega_{i})$ is made difficult as most of the mass of  $\nu_i$  is located at either the right or left side of the interval $\Omega$ which is not the case for its projection. 
The histograms displayed in  the right clumn  of Figure \ref{maps_name} correspond to the elements in the dataset that have a large Wasserstein distance with respect to the barycenter $\bar{\bnu}$. This explains why it is difficult to have good projected measures with GPCA. For elements in the dataset that are closest to $\bar{\bnu}$, the projected measures $  \exp_{\bar{\bnu}}(\Pi_{\spann (\tilde{u}_1)} \omega_{i})$ and $\exp_{\bar{\bnu}}(\Pi_{\spann (u^{\ast}_1)} \omega_{i})$ are much closer to $\nu_i$ and  for such elements, log-PCA and the iterative geodesic algorithm lead to similar results in terms of data projection.

To better estimate the extremal data in Figure \ref{maps_name}, a solution is to increase the support of the data to the interval $\Omega_{0} = [1850,2050]$, and to perform log-PCA and GPCA in the Wasserstein space $W_2(\Omega_{0})$. The results are reported in  Figure \ref{maps_name_increased}. In that case, it can be observed that both algorithms lead to similar results, and that a better projection is obtained for the extremal data. Notice that with this extended support, all the mappings $\tilde{T}_{i} = \text{id} + \Pi_{\spann (\{\tilde{u}_1\})} \omega_{i}$ obtained with log-PCA are optimal in the sense that they are non-decreasing with a range inside $\Omega_{0}$.

Finally, we display in  Figure \ref{gpca_name} and  Figure \ref{gpca_name_increased}  the results of the iterative and geodesic surface algorithms with either $\Omega = [1900,2013]$ or with data supported on the extended support $\Omega_{0} = [1850,2050]$. The projection of the data onto the first principal geodesic suggests that the distribution of a name is deeply dependent on the part of the century. The second geodesic express a popular trend through a spike effect. In Figure \ref{gpca_name}, the artefacts in the principal surface that are obtained with the iterative algorithm  at the end of the century, correspond to the fact that the projection of the data $\omega_i$ onto the surface spanned by the first two components is not ensured to belong to the set $V_{\bar{\bnu}}(\Omega)$.

\begin{figure}[ht!]
\centering
\includegraphics[scale=0.62]{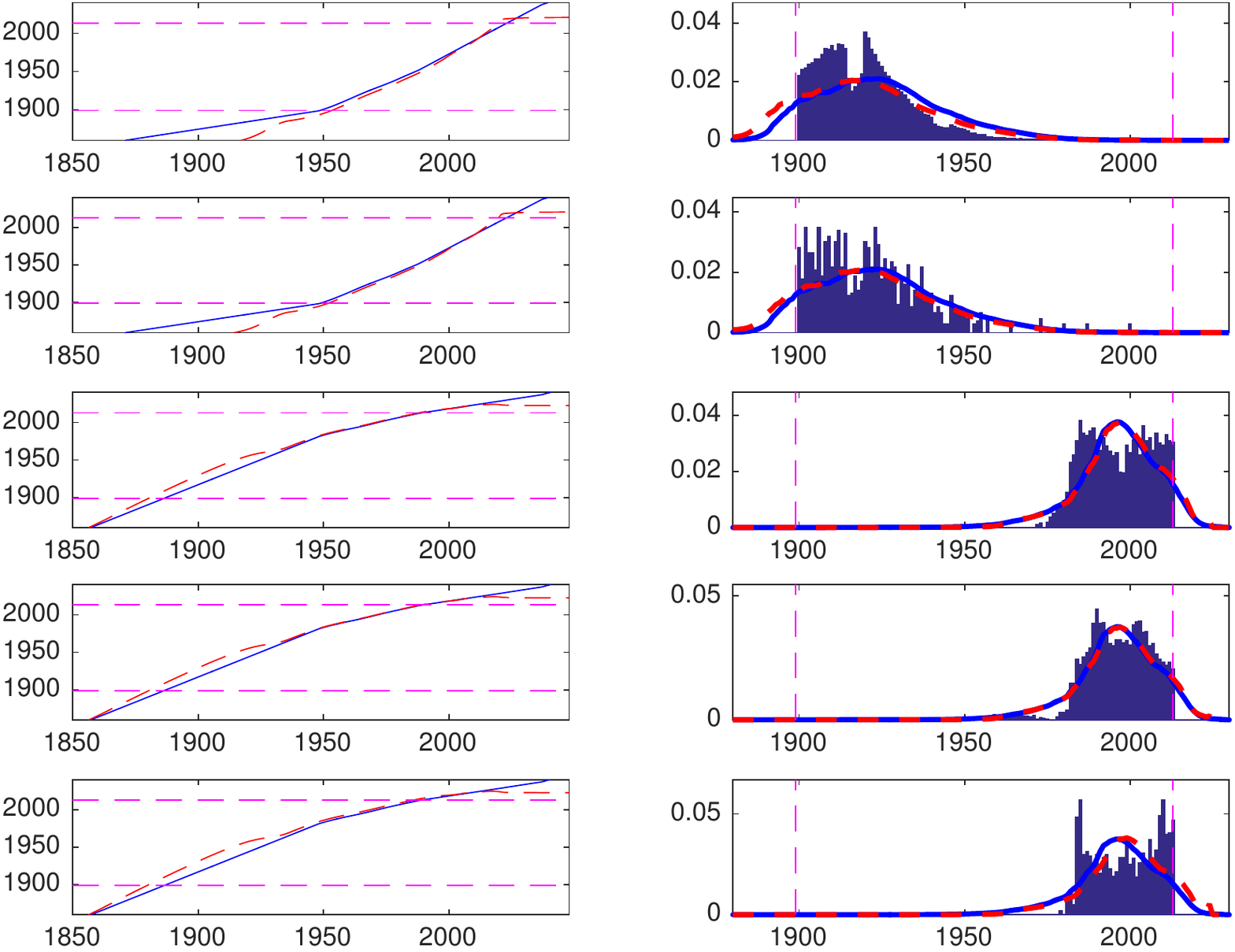} 
\caption{Children's first name at birth with extended support $\Omega_{0} = [1850,2050]$. (Left) The dashed red curves represent the mapping $\tilde{T}_{i} = \text{id} + \Pi_{\spann (\{\tilde{u}_1\})} \omega_{i}$  where $\omega_i = \log_{\bar{\bnu}}(\nu_i)$, and $\tilde{u}_1$ is the first principal direction in $\LLemp$ obtained via log-PCA. The blue curves are the mapping $T_{i} = \text{id} + \Pi_{\spann (\{u^{\ast}_{1}\})} \omega_{i}$, where  $u^{\ast}_{1}$ is the first principal direction in $\LLemp$ obtained via the iterative algorithm. (Right) The histogram stands for the pdf of measures $\nu_i$ that have a large Wasserstein distance with respect to the barycenter $\bar{\bnu}$. The red curves are the pdf of the projection   $\exp_{\bar{\bnu}}(\Pi_{\spann (\tilde{u}_1)} \omega_{i})$ with log-PCA, while the blue curves are the pdf of the projection $\exp_{\bar{\bnu}}(\Pi_{\spann (u^{\ast}_1)} \omega_{i})$ with GPCA.}
\label{maps_name_increased}
\end{figure}

\begin{figure}[ht!]
\centering
\includegraphics[width= 0.7\textwidth,height=0.46\textwidth]{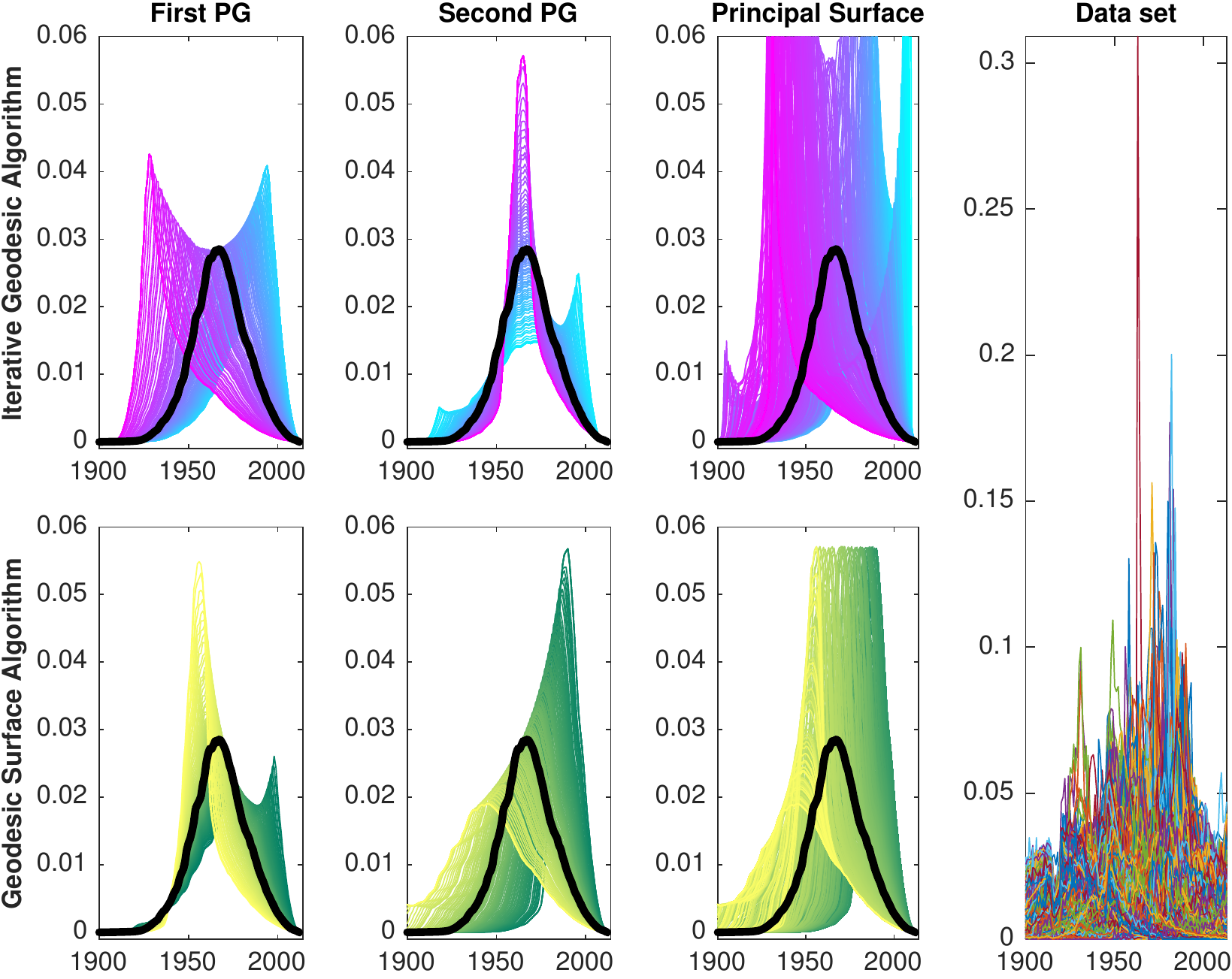} 
\caption{Children's first name at birth with support $\Omega = [1900,2013]$. The first row is the GPCA of the data set obtained with the iterative geodesic approach. The second row is the GPCA through the geodesic surface approach. The first (resp. second) column is the projection of the data into the first (resp. second) principal direction. The black curve is the density of the Wasserstein barycenter. Colors encode the progression of the pdf of principal geodesic components in $\WS$.}
\label{gpca_name}
\end{figure}

\begin{figure}[ht!]
\centering
\includegraphics[width= 0.7\textwidth,height=0.46\textwidth]{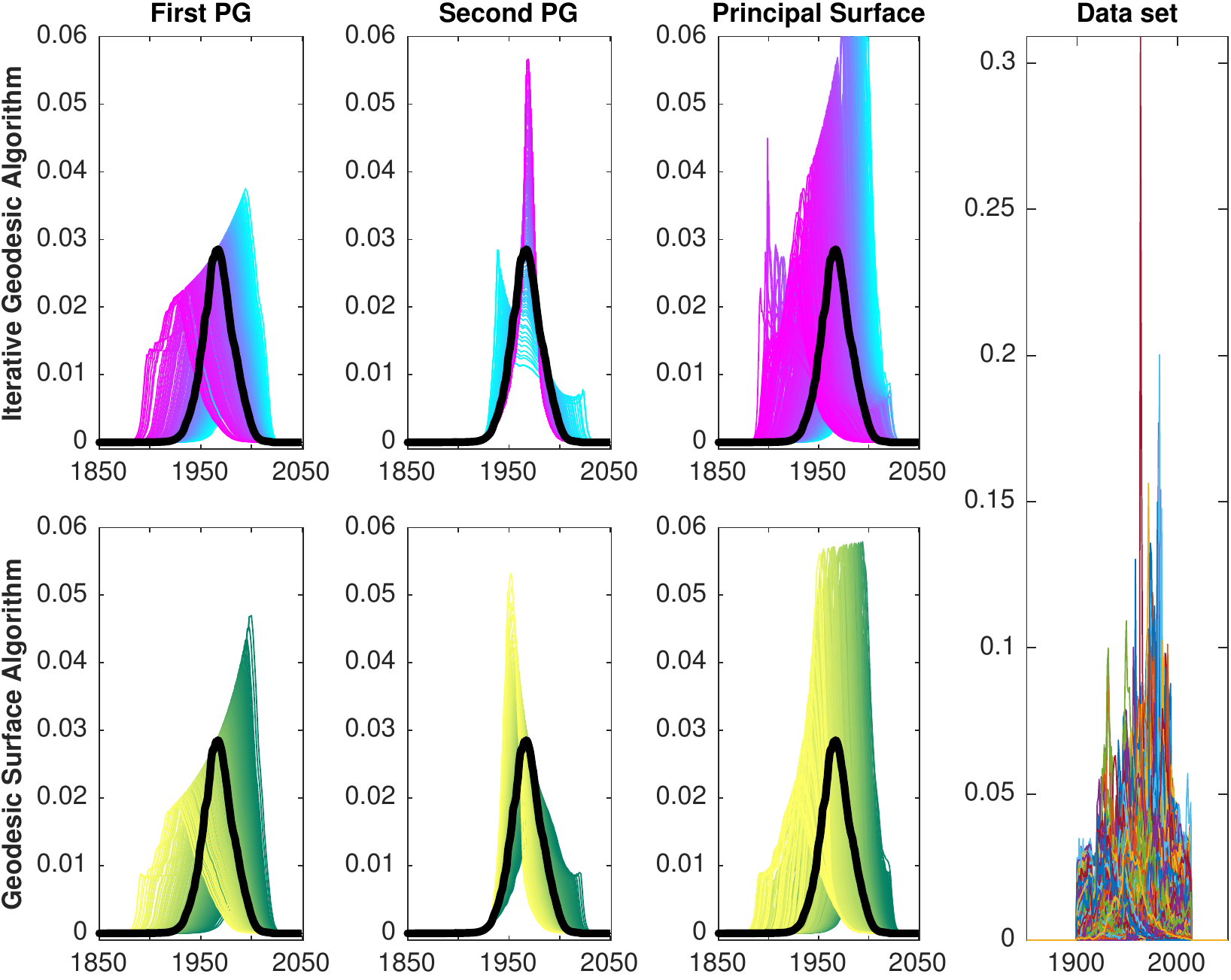} 
\caption{Children's first name at birth with extended support $\Omega_{0} = [1850,2050]$. The first row is the GPCA of the data set obtained with the iterative geodesic approach. The second row is the GPCA through the geodesic surface approach. The first (resp. second) column is the projection of the data into the first (resp. second) principal direction. The black curve is the density of the Wasserstein barycenter. Colors encode the progression of the pdf of principal geodesic components in $\WS$.}
\label{gpca_name_increased}
\end{figure}

\section{Extensions beyond $d > 1$ and some perspectives.} \label{sec:end}

We now briefly show that our iterative algorithm for finding principal geodesics can be adapted to the general case $d > 1$. This requires to take into account two differences with the one-dimensional case. First, the definition of the space $\VV$ in \ref{def:explog} relies on the explicit close-form formulae \eqref{eq:exp} for the computation of a solution to the optimal transport problem which is specific to the one-dimensional case. We must hence provide a more general definition of $\VV$. Second, the isometry property $(P1)$ does not hold for $d > 1$, so that Wasserstein distances cannot be replaced by the $L^2_{\bar{\bnu}}$ norm between log maps as in \eqref{eq:iterative_geodesic_isometry} and must be explicitly computed and differentiated.

\paragraph{Definition of $\VV$ in the general case} In the one dimensional case, $\VV$ is characterized in Proposition \ref{prop:exp} (P3) as the set of functions $v \in \LL$ such that $T:={\rm id}+v$ is $\mu_{r}$-almost everywhere non decreasing. A striking result by Brenier~\cite{brenier1991polar} is that, in any dimension, if $\mu_r$ does not give mass to small set, there exists an optimal mapping $T \in \LL$ between $\mu_r$ and any probability measure $\nu$, and $T$ is equal to the gradient of a convex function $u$ ie. $T = \nabla u$. Hence, we define the set $\VV$ as the set of functions $v \in \LL$ such that $id + v = \nabla u$ for an arbitrary convex function $u$.

In order to deal with the latter constraint, we note that this space of functions is equal to the space of functions $v \in \LL$ such that $div(v)\geq -1$. Indeed, assuming that $x+v=\nabla u$, then $u$ being a convex potential involves $div(\nabla u)\geq 0$ which is equivalent to $div(x+v)=div(v)+1\geq 0$.

\paragraph{General objective function} Without the isometry property (P1), the objective function $H(t_0,v)$ in \eqref{eq:optimCPCA3_1D} must be written with the explicit Wasserstein distance $d_W$,
\begin{equation}
H(t_{0},v) = \frac{1}{n} \sum_{i=1}^{n}  \min_{t_i\in[-1;1]} d_W^2(\nu_i, g_{t_i}(t_0, v))+\chi^{}_{V_{\bar{\bnu}}(\Omega)}((t_0-1)v)+\chi^{}_{V_{\bar{\bnu}}(\Omega)}((t_0+1)v),
\end{equation}
where $g_t(t_0, v) = (\text{id} + (t_0 + t)v)\# \bar{\bnu} \mbox{ for } t \in [-1;1]$ as defined in \eqref{eq:geodesic_parameterization}. 
Optimizing over both the functions $\mathbf{v} \in (\R^d)^N $ and the projection times $\mathbf{t}$, the discretized objective function to minimize is,\vspace{-0.1cm}
\begin{equation}
\label{eq:pb0_2D}\min_{\mathbf{v}\in\R^N}  \min_{\mathbf{t}\in\R^n}J(\mathbf{v},\mathbf{t}):=\underbrace{\sum_{i=1}^n d_W^2(\nu_i, g_{t_i}(t_0, \mathbf{v}))}_{F(\mathbf{v},\mathbf{t})}+\underbrace{\chi^{}_{S}(v)+\chi^{}_{E}(K\mathbf{v}) + \chi^{}_{D}(\mathbf{v}) + \chi^{}_{B^n_1}(\mathbf{t})}_{G(\mathbf{v},\mathbf{t})}.\vspace{-0.2cm}
\end{equation}
where $K$ is a discretized divergence operator, and $E = \{ \mathbf{z}\in \R^N~:~  \frac{-1}{t_0+1} \leqslant \mathbf{z} \leqslant \frac{1}{1-t_0} \}$, $D = \{ \mathbf{v}~:~ \id +(t_0 \pm 1)\mathbf{v} \in \Omega \}$ deals with the domain constraint  and $S$ deals with the orthogonality constraint w.r.t. to the preceding principal components.
As for the one-dimensional case, we minimize $J$ through the Forward-Backward algorithmas detailed in the appendix \ref{sec:prox_2D}.

Extension to higher dimensions is straightforward. However, considering that we have to discretize the support of the Wasserstein mean $\bar{\bnu}$, the approach becomes intractable for $d > 3$.

\subsection{Application to grayscale images}

We consider the MNIST dataset~\cite{lecun1998mnist} which contains grayscale images of handwritten digits. All the images have identical size $28 \times 28$ pixels. Each grayscale image, once normalized so that the sum of pixel grayscape values sum to one, can be interpreted as a discrete probability measure, which is supported on the 2D grid of size $28 \times 28$. The ground metric for the Wasserstein distance is then the 2D square Euclidean distance between the locations of the pixels of the two-dimensional grid. We compute the first principal components on 1000 images of each digit. Wasserstein barycenters, which are required as input to our algorithm, are approximated efficiently through iterative bregman projections as proposed in~\cite{benamou2015iterative}. We use the network simplex algorithm\footnote{\url{http://liris.cnrs.fr/~nbonneel/FastTransport/}} to compute Wasserstein distances. 

Figure \ref{fig:MNIST_GPCA}  displays the results obtained with our Forward-Backward algorithm (with $t_0$ set to $0$ for simplicity), and the ones  given by Log-PCA as  described in section \ref{sec:linear_pca}. These two figures are obtained by sampling the first principal components. We then use kernel smoothing to display the discrete probability measures back to the original grid and present the resulting grayscale image with an appropriate colormap. 

\begin{figure}[ht!]
\begin{center}
\begin{tabular}{cc}
\includegraphics[scale=0.3]{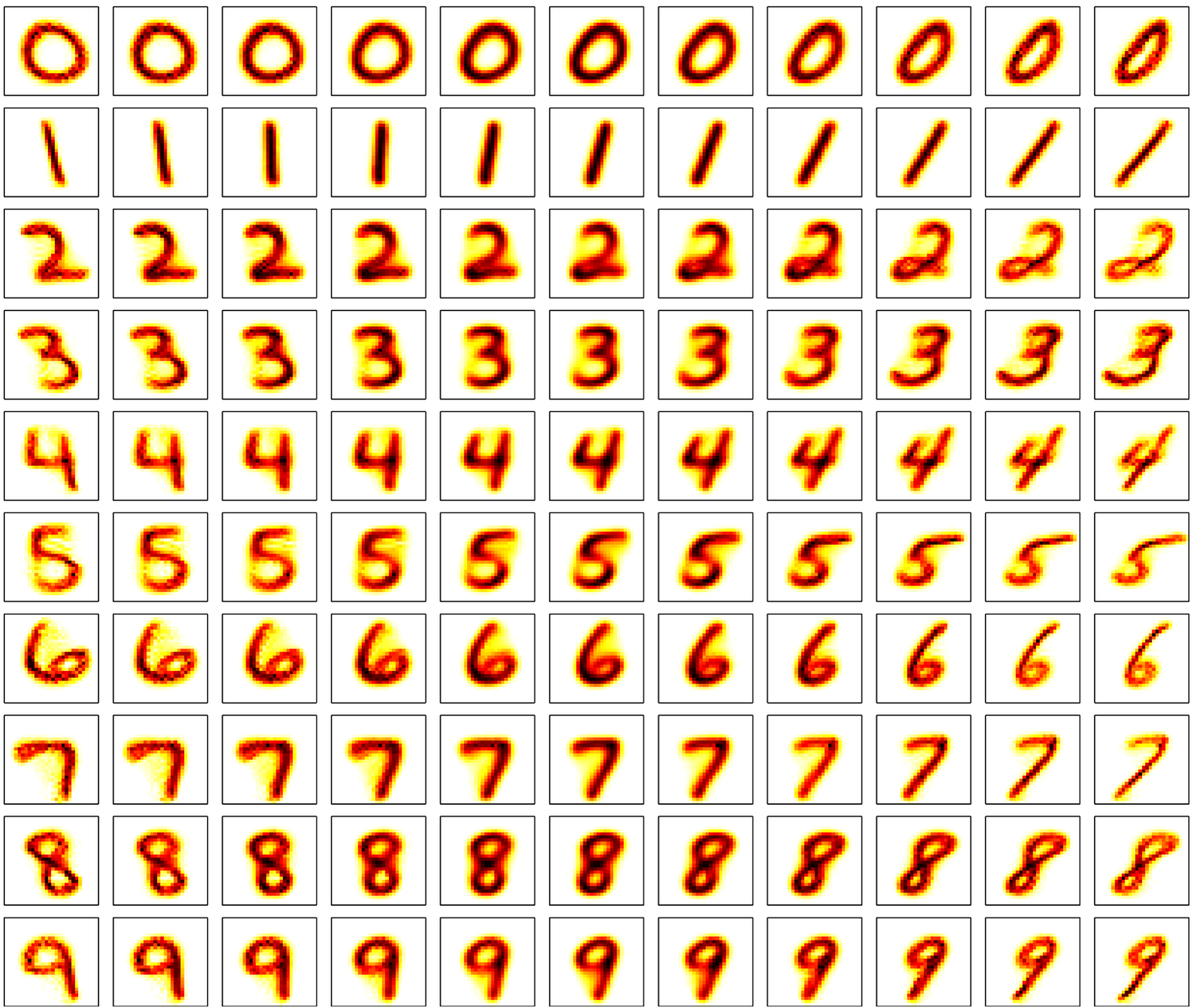}&\includegraphics[scale=0.3]{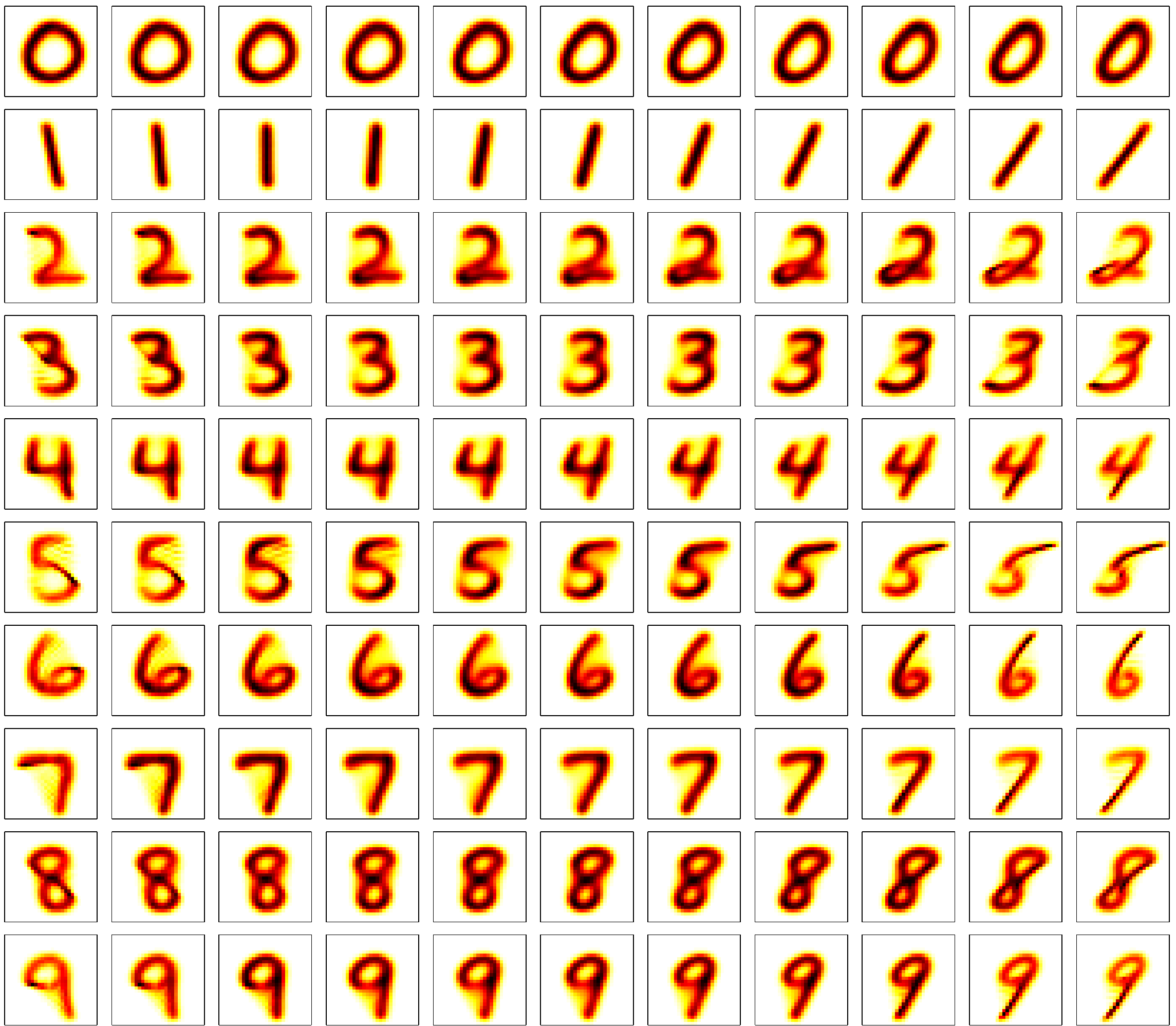} \\
GPCA&log-PCA
\end{tabular}
\caption{First principal geodesics for 1000 images of each digit from the MNIST dataset, computed through the proposed Forward-Backward algorithm (left) and log-PCA (right).}
\label{fig:MNIST_GPCA}
\end{center}
\end{figure}
Visually, both the Log-PCA and GPCA approaches capture well the main source of variability of each set of grayscape images ie each number. We observe variations in the slant of the handwritten digits for all digits, the most obvious case being digit '1'. As a principal component is parameterized by a whole velocity field on the support of the Wasserstein mean of the data, single principal components can capture more interesting patterns, such as changes in the shape of the '0' or the presence or absence of the lower loop of the '2'. From purely visual inspection, it is difficult to tell which approach, Log-PCA or GPCA, provides a ``better'' principal component. For this purpose we compute the reconstruction error of each digit. This reconstruction error is computed in the same way for both Log-PCA and GPCA principal components: We sample the principal components at many times $t$ and find for each image in a given dataset, the time at which the geodesic is the closest to the image sample. This provides an approximation of 
$	\min_{t\in[-1;1]} d_W^2(\nu_i, g_{t}(v))$ for each image $i=1,\ldots,n$, 
where $(g_{t})_{t \in [-1,1]}$ is the principal component. For the Log-PCA principal component, we take $\tilde{g}_{t} = (\text{id} +t 1.25 \lambda v)\#\bar{\bnu}$, where $\lambda$ is the eigenvalue corresponding to the first principal component. The $1.25$ factor is useful to consider a principal curve which goes through the whole range of the dataset. For the GPCA principal geodesic, we have $g^{\ast}_{t} = (\text{id} +tv)\#\bar{\bnu}$. The reconstruction errors are shown in Table \ref{table:RE}. We see that, for each digit, we obtain a better, i.e. smaller, reconstruction error when using the proposed Forward-Backward algorithm. This result is not surprising, since the reconstruction error is explicitly minimized through the Forward-Backward algorithm. As previously mentioned, Log-PCA rather computes linearized Wasserstein distances. In one-dimension, the isometry property (P1) states that these quantities are equal. In dimension two or larger, that property does not hold.
%
%

\begin{table}[h]
\begin{center}        
{\small                         
\begin{tabular}{|c|c|c|} \hline
\text{MNIST digit} & \text{Log-PCA RE}~ ($\cdot 10^3$)& \text{GPCA RE}~ ($\cdot 10^3$)\\ \hline
0 &  $2.0355$ & $\mathbf{1.9414}$ \\
1 &  $3.1426$ & $\mathbf{1.0289}$ \\
2 &  $3.4221$ & $\mathbf{3.3575}$ \\
3 &  $2.6528$ & $\mathbf{2.5869}$ \\
4 &  $2.8792$ & $\mathbf{2.8204}$ \\
5 &  $2.9391$ & $\mathbf{2.9076}$ \\
6 &  $2.1311$ & $\mathbf{1.9864}$ \\
7 &  $4.7471$ & $\mathbf{2.8205}$ \\
8 &  $2.0741$ & $\mathbf{2.0222}$ \\
9 &  $1.9303$ & $\mathbf{1.8728}$ \\
\hline
\end{tabular}}
\caption{\label{table:RE}Reconstruction Errors (RE) computed on 1000 sample images of each digit of the MNIST dataset. (center) Reconstruction error w.r.t. the first principal component computed with the Log-PCA algorithm. (right) Reconstruction error w.r.t. the first principal geodesic computed with the proposed Forward-Backward algorithm.\vspace{-0.6cm}}   
\end{center}                                
\end{table}

\subsection{Discussion}

The proposed Forward-Backward algorithm minimizes the same objective function as defined in~\cite{NIPS2015_5680}. The first difference with the algorithm provided~\cite{NIPS2015_5680} is that we take gradient steps with respect to both $\mathbf{v}$ and $\mathbf{t}$, while the latter first attempts to find the optimal $t$ (by sampling the geodesics at many time $t$), before taking a gradient step of $\mathbf{v}$. Our approach reduces the cost of computing a gradient step by one order of magnitude. Secondly,~\cite{NIPS2015_5680} relied on barycentric projections to preserve the geodesicity of the principal curves in between gradient steps. That heuristic does not guarantee a decrease in the objective after a gradient step. Morever, the method in~\cite{NIPS2015_5680} considered two velocity fields $\mathbf{v}_1,\mathbf{v}_2$ rather than a single $\mathbf{v}$ since  the optimality of both $\mathbf{v}$ and $-\mathbf{v}$ could not be preserved through the barycentric projection.

When considering probability measures over high dimensional space ($d > 3$), our algorithm becomes intractable since we need to discretize the support of the Wasserstein mean of the data with a regular grid, while the approach of~\cite{NIPS2015_5680} is still tractable since an arbitrary support for the Wasserstein mean is used. A remaining challenge for computing principal geodesics in the Wasserstein space is then to propose an algorithm for GPCA which is still tractable in higher dimensions while not relying on barycentric projections.


\appendix
\section{Dimension $d=1$} \label{sec:annexe}
We here detail the application of Algorithm \eqref{algoFB} to the iterative GPCA procedure that consists in solving the problem \eqref{eq:pb0}:
 \begin{equation*}\min_{\mathbf{v}\in\R^N}  \min_{\mathbf{t}\in\R^n}J(\mathbf{v},\mathbf{t}):=\underbrace{\sum_{i=1}^n  \| \mathbf{w}_i-(t_0+t_i)\mathbf{v} \|^{2}_{\bar{\bnu}}}_{F(\mathbf{v},\mathbf{t})}+\underbrace{\chi^{}_S(\mathbf{v})  +\chi^{}_{D}(\mathbf{\mathbf{v}})+\chi^{}_{\spaceK}(K\mathbf{v})+ \chi^{}_{B^n_1}(\mathbf{t})}_{G(\mathbf{v},\mathbf{t})}.
\end{equation*}

\subsection{Lispschitz constant of $\nabla F$} \label{sec:LipschitzCst}
Let us now look at the Lipschitz constant  of $\nabla F(\mathbf{v},\mathbf{t})$ on the  restricted acceptable set $ D\times  B^n_1$.
We first denote as $\mathcal{H}$ the hessian matrix (of size $(N+n)\times(N+n)$) of the $\mathcal{C}^2$ function $F(X)$. We know that if  the spectral radius of $\mathcal{H}$  is bounded by a scalar value $M$, i.e. $\rho(\mathcal{H})\leq M$, then $\nabla F$ is a Lipschitz continuous function with constant $M$.
Hence, we  look at the eigenvalues of the Hessian matrix of  $F=\sum_{i=1}^n \sum_{j=1}^N\bar{\bfun}_{n}(x_j) (w_i^j-(t_0+t_i) v_j)^2$ that is
\begin{equation*}
\frac{\partial^2 F}{\partial t_i^2}=\sum_{j=1}^N 2 v_j^2\bar{\bfun}_{n}(x_j),\quad
\frac{\partial^2 F}{\partial v_j^2}=\sum_{i=1}^n 2 (t_0+t_i)^2\bar{\bfun}_{n}(x_j),\quad
\frac{\partial^2 F}{\partial t_i\partial v_j}=2\bar{\bfun}_{n}(x_j)(2(t_0+t_i)v_j-w_i^j)\\
\end{equation*}
and $\frac{\partial^2 F}{\partial t_i\partial t_{i'}}=\frac{\partial^2 F}{\partial v_j\partial v_{j'}}=0$, for all $i\neq i'$ or $j\neq j'$.
%
Being $\{\mu_k\}_{k=1}^{n+N}$ the eigenvalues of $\mathcal{H}$, we have $\rho(\mathcal{H})=\max_k |\mu_k|\leq \max_k \sum_l |\mathcal{H}_{kl}|$.
We denote as  $f_\infty=\max_j|\bar{\bfun}_{n}(x_j)|$ and likewise $w_\infty=\max_{i,j}|w_i^j|$. Since $|t_0|<1$,  $t_i^2\leq1$, $\forall \mathbf{t}\in B^n_1$ and $v_j^2\leq\alpha^2=(b-a)^2$, $\forall \mathbf{v}\in D$, by defining $\gamma=2(1+|t_0|)\alpha+w_\infty$, we  thus have 
 \begin{equation}\label{lip_const}\rho(\mathcal{H})\leq 2f_\infty \max\left\{ n\alpha^2+N\gamma, n\gamma+N(1+|t_0|)^2\right\}:=M.\end{equation}
 \subsection{Computing $\textrm{Prox}_{\tau G}$}\label{sec:proxG} In order to implement the algorithm \eqref{algoFB}, we finally need to compute the proximity operator of $G$ defined as:
$$(\mathbf{v}^*,\mathbf{t}^*)=\Prox_{\tau G}(\tilde{\mathbf{v}}, \tilde {\mathbf{t}})=\uargmin{\mathbf{v},\mathbf{t}} \frac1{2\tau}(||\mathbf{v}-\tilde{ \mathbf{v}}||^2+||\mathbf{t}-\tilde {\mathbf{t}}||^2)+\chi^{}_S(\mathbf{v})+\chi^{}_{D}(\mathbf{v})+\chi^{}_{\spaceK}(K\mathbf{v}) + \chi^{}_{B^n_1}(\mathbf{t}).$$
This problem can be solved independently  on $\mathbf{v}$ and $\mathbf{t}$.
For $\mathbf{t}$, it can be done pointwise as 
$t^*_i=\argmin_{t_i} \frac1{2\tau}||t_i-\tilde t_i||^2+\chi^{}_{B^1_1}(t_i)=\Proj_{[-1;1]}(\tilde t_i).$
Unfortunately, there is no closed form expression of the proximity operator for the component $\mathbf{v}$. It requires to solve the following  intern optimization problem at each extern iteration $(\ell)$ of the algorithm \eqref{algoFB}:
 \begin{equation}\label{intern_pb}\mathbf{v}^*=\uargmin{\mathbf{v}} \frac1{2\tau}||\mathbf{v}-\tilde{\mathbf{v}}||^2+\chi^{}_S(\mathbf{v})+\chi^{}_{D}(\mathbf{v})+\chi^{}_{\spaceK}(K\mathbf{v}),\end{equation}
where, to avoid confusions, we  denote by $\mathbf{v}$ the variable that is optimized within the intern optimization problem \eqref{intern_pb}.
\begin{rem}
The Lipschitz constant of $\nabla F(\mathbf{v},\mathbf{t})$ in \eqref{lip_const} relies independantly on $\mathbf{v}$ and $|t_0|$, thus we can choose the optimal gradient descent step $\tau$ for $\mathbf{v}^*$ and $\mathbf{t}^*$.
\end{rem}
\paragraph{Primal-Dual reformulation} Using duality (through Fenchel transform), one has:
 \begin{align}
&\min_{\mathbf{v}\in\R^N}\frac1{2\tau}||\mathbf{v}-\tilde {\mathbf{v}}||^2+\chi^{}_S(\mathbf{v})+\chi^{}_{D}(\mathbf{v})+\chi^{}_{\spaceK}(K\mathbf{v})\nonumber\\
=&\min_{\mathbf{v}\in\R^N}\max_{\mathbf{z}\in\R^N}\frac1{2\tau}||\mathbf{v}-\tilde {\mathbf{v}}||^2+\chi^{}_S(\mathbf{v})+\chi^{}_{D}(\mathbf{v})+\langle K\mathbf{v},\mathbf{z}\rangle -\chi^*_{\spaceK}(\mathbf{z}),\label{eq:pd}
\end{align}
where $\mathbf{z}=\{\mathrm{z}_j\}_{j=1}^N\in \R^N$ is a dual variable and $\chi^*_{\spaceK}=\sup_\mathbf{v} \langle \mathbf{v},\mathbf{z}\rangle -\chi^{}_{\spaceK}(\mathbf{v})$ is the convex conjugate of $\chi^{}_{\spaceK}$ that reads:
 \begin{equation*}
(\chi^*_{\spaceK}(\mathbf{z}))_j=\left\{\begin{array}{ll} 
-\mathrm{z}_j/(1+t_0)
&\textrm{if } \mathrm{z}_j\leq 0,\\
\mathrm{z}_j/(1-t_0)&\textrm{if } \mathrm{z}_j> 0.\end{array}\right.
\end{equation*}
Hence, one can use the Primal-Dual algorithm proposed in \cite{CP14} to solve the problem \eqref{eq:pd}. For two parameters $\sigma, \theta>0$ such that $||K||^2\leq \frac1\sigma (\frac1 \theta-\frac 1\tau) $ and given $\mathbf{v}^0,\bar{ \mathbf{v}}^0,\mathbf{z}^0\in \R^N$,  the algorithm is:
 \begin{equation}\label{PD-prox}\left\{\begin{array}{lll}
\mathbf{z}^{(m+1)}&=&\Prox_{\sigma \chi^*_{\spaceK}}(\mathbf{z}^{(m)}+\sigma K\bar {\mathbf{v}}^{(m)})\\
\mathbf{v}^{(m+1)}&=&\Prox_{\theta (\chi_{D}+\chi^{}_S)}(\mathbf{v}^{(m)}-\theta (K^* \mathbf{z}^{(m+1)} +\frac1\tau (\mathbf{v}^{(m)}-\tilde {\mathbf{v}}) )\\
\bar {\mathbf{v}}^{(m+1)}&=&2\mathbf{v}^{(m+1)}-\mathbf{v}^{(m)}\\
\end{array}\right.\end{equation}
where $K^*$ is defined as $\langle K\mathbf{v},\mathbf{z}\rangle= \langle \mathbf{v},K^*\mathbf{z}\rangle$. Using the operator $K$ defined in \eqref{op:K}, we thus have:
\begin{equation}\label{op:Kt}(K^*\mathbf{z})_j =\left\{\begin{array}{ll}-\mathrm{z}_{1}/\Delta_1&\textrm{if }j=1\\  \mathrm{z}_{j-1}/\Delta_{j-1}-\mathrm{z}_{j}/\Delta_j&\textrm{if } 1<j<N.\\
\mathrm{z}_{N-1}/\Delta_{N-1}&\textrm{if } j=N,
\end{array}\right.\end{equation} 
where $\Delta_j=x_{j+1}-x_j$.  We have that $||K||^2=\rho(K^*K)$,  the largest eigenvalue of $K^*K$. With the discrete operators \eqref{op:K} and \eqref{op:Kt}, $\rho(K^*K)$  can be bounded by 
 \begin{equation}\label{norme_K}
 \delta^2=2\max_j (1/\Delta_j^2 + 1/\Delta_{j+1}^2).\end{equation}
One can therefore for instance take  $\sigma=\frac1\delta$ and $\theta=\tau/(1+\delta\tau)$.
\paragraph{Proximity operators in \eqref{PD-prox}} 
The proximity operator of $\chi_{D}+\chi_S$ is obtained as:
 \begin{equation}\label{prox:v}(\Prox_{\theta (\chi_{D}+\chi^{}_S)}(\mathbf{v}))_j=(\Proj_{D\cap S}(\mathbf{v}))_j=\Proj_{[m_j;M_j]}\left( \left(\mathbf{v}-\sum_{l=1}^{k-1}\frac{\langle \mathbf{u}^{}_l,\mathbf{v}\rangle_{\bar{\bnu}}}{||\mathbf{u}^{}_l||_{\bar{\bnu}}^2} \mathbf{u}^{}_l  \right)_j\,\right),\end{equation}
since projecting onto $D\cap S$ is equivalent to first project onto the orthogonal of $\spann({\U}^{k-1})$ and then onto $D$.
One can finally show that the proximity operator of $\chi^*_{\spaceK}$ 
can be computed pointwise as:
 \begin{equation}\label{prox:z}(\Prox_{\sigma \chi^*_{\spaceK}}(\mathbf{z}))_j=
 \left\{\begin{array}{ll}
\mathrm{z}_j-\sigma/(1-t_0)&\textrm{if }  \mathrm{z}_j>\sigma/(1-t_0)\\
\mathrm{z}_j+\sigma/(1+t_0)&\textrm{if }  \mathrm{z}_j<-\sigma/(1+t_0)\\
0&\textrm{otherwise}.
\end{array}\right.
\end{equation}

\subsection{Algorithms for GPCA}
\label{sec:GPCA_algorithm}
Gathering all the previous elements, we can finally find a critical point of the non-convex problem \eqref{eq:pb0} using the 
Forward-Backward (FB) framework \eqref{algoFB}, as detailed in Algorithm \ref{algo}.

\begin{algorithm}[H]
\caption{\label{algo} Resolution with FB of problem \eqref{eq:pb0}: $\min_{\mathbf{v},\mathbf{t}} F(\mathbf{v},\mathbf{t})+G(\mathbf{v},\mathbf{t})$}
\begin{algorithmic}
 \REQUIRE $\boldsymbol{w}_i\in \R^N$ for $i=1\cdots n$,  $\mathbf{u}^{}_1,\cdots \mathbf{u}^{}_{k-1}$,  $t_0\in]-1;1[$, $\alpha=(b-a)>0$, 
 $\eta>0$, $\delta>0$ (defined in \eqref{norme_K})  and $M>0$ (defined in \eqref{lip_const}).
\STATE Set $(\mathbf{v}^{(0)},\mathbf{t}^{(0)})\in D\times B^n_1$
\STATE Set $\tau<1/M$, $\sigma=1/\delta$ and $\theta=\tau/(1+\delta\tau)$.
\STATE \%Extern loop:
\WHILE {$||\iter{\mathbf{v}}-\ter{\mathbf{v}}||/||\ter{\mathbf{v}}||>\eta$} 
\STATE \% FB on $\mathbf{t}$ with $\mathbf{t}^{(\ell+1)}=Prox_{\tau G}(\mathbf{t}^{(\ell)}-\tau \nabla F(\mathbf{v}^{(\ell)},\mathbf{t}^{(\ell)}))$: 
\STATE  $\iiter{t_i}=\Proj_{[-1;1]}\left(\iter{t_i}-\tau\sum_{j=1}^{N} \iter{v_j} \bar{\bfun}_{n}(x_j)\left((t_0+\iter{t_i})\iter{v_j}-w^j_i\right)\right)$
\STATE \% Gradient descent on $\mathbf{v}$ with $\tilde{\mathbf{v}}=\mathbf{v}^{(\ell)}-\tau \nabla F(\mathbf{v}^{(\ell)},\mathbf{t}^{(\ell)})$:
\STATE  $\tilde v_j=\iter{v_j}-\tau\bar{\bfun}_{n}(x_j)\sum_{i=1}^{n} (t_0+\iter{t_i}) \left((t_0+\iter{t_i}) \ \iter{v_j}-w^j_i\right)$\vspace{0.1cm}
\STATE \%Intern loop for $\mathbf{v}^{(\ell+1)}=Prox_{\tau G}(\tilde{\mathbf{v}})$:
\STATE Set $\mathbf{z}^{(0)}\in \spaceK$, $\mathbf{v}^{(0)}=\tilde {\mathbf{v}}$, $\bar {\mathbf{v}}^{(0)}=\tilde {\mathbf{v}}$
\WHILE{$||\mathbf{v}^{(m)}-\mathbf{v}^{(m-1)}||/||\mathbf{v}^{(m-1)}||>\eta$}
\STATE $\mathbf{z}^{(m+1)}=\Prox_{\sigma \chi^*_{\spaceK}}\left(\mathbf{z}^{(m)}+\sigma K\bar {\mathbf{v}}^{(m)}\right)$\hfill(using \eqref{prox:z})
\STATE $\mathbf{v}^{(m+1)}=\Prox_{\theta (\chi_{D}+\chi^{}_S)}\left(\mathbf{v}^{(m)}-\theta (K^* \mathbf{z}^{(m+1)} +\frac1\tau (\mathbf{v}^{(m)}-\tilde {\mathbf{v}}) \right)$\hfill(using \eqref{prox:v})
\STATE $\bar {\mathbf{v}}^{(m+1)}=2\mathbf{v}^{(m+1)}-\mathbf{v}^{(m)}$
\STATE $m:=m+1$
\ENDWHILE 
\STATE $\iiter{\mathbf{v}}=\mathbf{v}^{(m)}$
\STATE $\ell:=\ell+1$
\ENDWHILE
\RETURN $\mathbf{u}^{}_k=\iter{\mathbf{v}}$
\end{algorithmic}
\end{algorithm}

\paragraph{Geodesic surface approach}
In order to solve the problem \eqref{eq:surf0}, 
%
we follow the same steps as in the section \ref{sec:LipschitzCst}-\ref{sec:proxG}. 
First we obtain the Lipchitz constant of the function $\tilde F$ by the same tricks as in the iterative algorithm. 
Then, since the constraints' problem in $G'$ are separable, we can compute each component $\mathbf{v_k}$ and each $\mathbf{\alpha_{i}^{\pm}}$ independantly. The only difference with the iterative algorithm concerns the proximal operator of the function $\chi_A$, which is the projection into the simplex of $\R^{2K}$.

\section{Dimension $d = 2$}\label{sec:proxG_2D} 
\label{sec:prox_2D}
We now show how to generalize the algorithm to the two-dimensional case. 
\paragraph{Gradients of F.} We write $X = (x_1, \cdots, x_N) \in (\mathbb{R}^2)^N$ the discretized support of $\bar{\bnu}$, $Z_t = (x_1 +(t_0+t)v_1, \cdots, x_N +(t_0+t)v_N)$ the support $g_{t}(t_0, \mathbf{v})$, the geodesic sampled at time $t$. Let $P^*$ be an optimal transport plan between $\bar{\bnu}$ and $g_{t}(t_0,\mathbf{v})$. The function $F(\mathbf{v},\mathbf{t})$ is differentiable almost everywhere. Gradients can be computed in the same fashion as \cite{NIPS2015_5680} to obtain,
\begin{equation}
	\label{eq:gradients_general_d}
	\nabla _\mathbf{v} F = 2 \sum_{i=1}^n (t_0+t_i) (Z_{t_i} - XP^{*T}\text{diag}(1/\bar{\bfun}_{n})), ~~~~	\nabla _{t_i} F = 2\dotprod{Z_{t_i}\text{diag}(\bar{\bfun}_{n})}{\mathbf{v}} -2\dotprod{P^*}{\mathbf{v}^TX},
\end{equation}
\paragraph{Proximal operator of G.} 
The only difference between the one-dimensional case and the two-dimensional case considered here concerns the projection step of $\mathbf{v}$,
\begin{equation}\label{intern_pb_2D}
	 \mathbf{v}^*=\uargmin{\mathbf{v}} \frac1{2\tau}||\mathbf{v}-\tilde{\mathbf{v}}||^2+\chi^{}_S(\mathbf{v})+\chi^{}_{D}(\mathbf{v})+\chi^{}_{\spaceK}(K\mathbf{v}),
 \end{equation}

%

%

\paragraph{Primal-Dual reformulation} As for the on-dimensional case, one has,
 \begin{align}
&\min_{\mathbf{v}\in\R^N}\frac1{2\tau}||\mathbf{v}-\tilde {\mathbf{v}}||^2+\chi^{}_S(\mathbf{v})+\chi^{}_{D}(\mathbf{v})+\chi^{}_{\spaceK}(K\mathbf{v})\nonumber\\
=&\min_{\mathbf{v}\in\R^N}\max_{\mathbf{z}\in\R^N}\frac1{2\tau}||\mathbf{v}-\tilde {\mathbf{v}}||^2+\chi^{}_S(\mathbf{v})+\chi^{}_{D}(\mathbf{v})+\langle K\mathbf{v},\mathbf{z}\rangle -\chi^*_{\spaceK}(\mathbf{z}),\label{eq:pd_2D}
\end{align}
where $\mathbf{z}=\{\mathrm{z}_j\}_{j=1}^N\in \R^N$ is a dual variable and $\chi^*_{\spaceK}=\sup_\mathbf{v} \langle \mathbf{v},\mathbf{z}\rangle -\chi^{}_{\spaceK}(\mathbf{v})$ is the convex conjugate of $\chi^{}_{\spaceK}$. This can be solve with the same iterative steps as described in \ref{sec:proxG},
 \begin{equation}\label{PD-prox_2D}\left\{\begin{array}{lll}
\mathbf{z}^{(m+1)}&=&\Prox_{\sigma \chi^*_{\spaceK}}(\mathbf{z}^{(m)}+\sigma K\bar {\mathbf{v}}^{(m)})\\
\mathbf{v}^{(m+1)}&=&\Prox_{\theta (\chi_{D}+\chi^{}_S)}(\mathbf{v}^{(m)}-\theta (K^* \mathbf{z}^{(m+1)} +\frac1\tau (\mathbf{v}^{(m)}-\tilde {\mathbf{v}}) )\\
\bar {\mathbf{v}}^{(m+1)}&=&2\mathbf{v}^{(m+1)}-\mathbf{v}^{(m)}\\
\end{array}\right.\end{equation}
Here the definition of the divergence operator $K$ and the transpose of the divergence operator $K^*$ are specific to the dimension. For $d = 2$, 
with a regular grid discretizing $\Omega$ in $M\times N$ points, we take
$$K^T \mathbf{z} =-\nabla \mathbf{z} =-\begin{bmatrix}\partial_x^+\mathbf{z}\\\partial_y^+\mathbf{z}\end{bmatrix},$$
with $$\partial_x^+\mathbf{z}(i,j)=\left\{\begin{array}{ll}\mathbf{z}(i+1,j)-\mathbf{z}(i,j)&if\, i<M\\
0&otherwise,\end{array}\right.
$$
$$\partial_y^+\mathbf{z}(i,j)=\left\{\begin{array}{ll}\mathbf{z}(i,j+1)-\mathbf{z}(i,j)&if\, j<N\\
0&otherwise\end{array}\right.
$$
so that $$K \mathbf{u} = K\begin{bmatrix}\mathbf{u}_x\\\mathbf{u}_y\end{bmatrix}=\partial_x^-\mathbf{u}_x+\partial_y^-\mathbf{u}_y,$$
with
$$\partial_x^-\mathbf{u}(i,j)=\left\{\begin{array}{ll}\mathbf{u}(i,j)-\mathbf{u}(i-1,j)&if\, 1<i<M\\
\mathbf{u}(i,j)&if\,i=1\\
-\mathbf{u}(i-1,j)&if\,i=M.\end{array}\right.
$$
To ensure convergence of \ref{PD-prox_2D}, one can take $ 1/\sigma . (1/\theta-1/\tau)= ||K||^2$. See \cite{CP15,LP15} for more details.
Since we have $||K||^2=8$, the parameters can be taken as $\sigma=1/4$ and $\theta=\tau/(1+2\tau)$.

\bibliographystyle{alpha}
\bibliography{HistPCAW2}

\end{document}